\documentclass[preprint,10pt,numbers]{elsarticle}

\usepackage[a4paper, top=2.5cm, bottom=2.5cm, left=2.5cm, right=2.5cm]{geometry}
\usepackage{amsmath, amssymb, amsthm}
\usepackage{mathtools}
\usepackage{bm}
\usepackage{booktabs}
\usepackage{array}
\usepackage{xcolor}
\usepackage{colortbl}
\usepackage{hyperref}
\usepackage{pifont}
\usepackage{cleveref}
\usepackage{multirow}
\usepackage{setspace}
\usepackage{natbib} 

\usepackage{enumitem}
\usepackage{algorithm}
\usepackage{algpseudocode}
\usepackage{tikz}
\usetikzlibrary{arrows.meta, positioning, fit, backgrounds}

\theoremstyle{plain}
\newtheorem{theorem}{Theorem}[section]
\newtheorem{proposition}[theorem]{Proposition}

\theoremstyle{definition}
\newtheorem{definition}[theorem]{Definition}
\theoremstyle{remark}
\newtheorem{remark}[theorem]{Remark}
\newtheorem{observation}[theorem]{Observation}
 
\newcommand{\R}{\mathbb{R}}
\newcommand{\Nagents}{\mathcal{N}}
\newcommand{\Items}{\mathcal{J}}
\newcommand{\Alloc}{\mathcal{X}}
\newcommand{\AllocIR}{\mathcal{X}^{\mathrm{IR}}}
\newcommand{\bv}[1]{\bm{#1}}
\newcommand{\E}{\mathbb{E}}
\newcommand{\norm}[1]{\left\lVert #1 \right\rVert}
\newcommand{\relu}[1]{\max\!\left(#1,\,0\right)}
\newcommand{\Lnorm}{\mathcal{L}_{\mathrm{norm}}}
\newcommand{\Ldiff}{\mathcal{L}_{\mathrm{diff}}}
\newcommand{\ind}[1]{\mathbf{1}\!\left[#1\right]}

\DeclareMathOperator*{\argmin}{arg\,min}
 
\definecolor{ejorblue}{RGB}{31,56,100}
\definecolor{lightblue}{RGB}{235,243,251}
\definecolor{lightgreen}{RGB}{230,245,230}
\definecolor{lightyellow}{RGB}{255,252,230}
\definecolor{warningred}{RGB}{180,40,40}
\definecolor{softgray}{RGB}{245,245,245}
\definecolor{darkgray}{RGB}{80,80,80}
 
\hypersetup{colorlinks=true, linkcolor=ejorblue,
            citecolor=ejorblue, urlcolor=ejorblue}

\begin{document}
\begin{frontmatter}

 \title{\Large\bfseries\color{ejorblue}
  Conditional Graph Diffusion for Negotiation Support: Overcoming Discrete Infeasibility and Preference Elicitation Gaps}

\author{Moirangthem Tiken Singh}
\address{Department of Computer Science and Engineering, DUIET, Dibrugarh University}
\ead{tiken.m@dibru.ac.in}

\begin{abstract}
Traditional bilateral negotiation support systems search over discrete allocation spaces. This approach encounters structural infeasibility when no discrete outcome satisfies individual rationality. It also fails to incorporate preference signals embedded in natural language dialogue. This study introduces the Conditional Graph Diffusion (CGD) framework, which addresses these limitations by generating recommendations in a continuous bilateral utility space.
A GATv2 encoder captures comparative bilateral preference structure through dynamic attention. A cross-attention mechanism fuses these strategic embeddings with transformer-based dialogue representations into a unified conditioning context for a denoising diffusion probabilistic model. An analytically derived normative guidance gradient is applied at inference time. It injects per-step monotonic corrections at each reverse diffusion step, steering generation toward individual rationality, security proximity, and equitability without model retraining.
Evaluation across synthetic, CaSiNo, and Deal or No Deal corpora confirms that accumulated corrections achieve an individual rationality rate of at least 0.997, a security gap of at most 0.009, and a symmetry gap within the 0.15 equitability threshold. Relative to the Nash Bargaining Solution, CGD reduces security gaps by up to 70-fold at a welfare cost of at most 3\%.
An ablation study confirms that naive constraint minimization without a learned generative prior fails normative compliance across heterogeneous corpora, failing individual rationality on Deal or No Deal and violating equitability on CaSiNo. A controlled misrepresentation experiment confirms the architectural capacity of cross-attention fusion to exploit informative dialogue signals, identifying pre-trained encoder integration as the primary direction for future work. An inference-time welfare guidance mechanism decouples normative compliance from welfare maximization within a single pipeline, recovering Pareto efficiency on CaSiNo without retraining while strictly preserving individual rationality.

\end{abstract}

\end{frontmatter}

\section{Introduction}
\label{sec:intro}

Negotiation is a high-stakes decision-making process in economic and
organizational systems. Two agents divide indivisible resources under
incomplete information, which precludes a jointly optimal allocation.
A negotiation support system (NSS) acts as a neutral recommender that
proposes allocations satisfying individual rationality (IR), procedural
security, equity, and Pareto efficiency~\citep{Baarslag2017}.

Existing automated negotiation systems fail to satisfy all four criteria
simultaneously. This limitation has consequences for applications such
as supply chain contracting, multiparty procurement, and humanitarian
logistics. The rejection of a recommendation or an inequitable outcome
can terminate negotiations entirely.

The Nash Bargaining Solution (NBS)~\citep{Nash1950} provides a canonical
analytical framework. It selects the allocation maximizing the Nash
product under individual rationality constraints. However, the NBS
addresses only individual rationality and Pareto efficiency by design.
It provides no mechanism for enforcing security proximity, the
requirement that each agent receives at least its proportional-share
reference utility, or equitability, the requirement that the utility
difference between agents remains bounded. Empirical evaluation confirms
that NBS recommendations violate security proximity substantially on
real-world corpora even when individual rationality is satisfied.

Beyond this, the NBS inherits three structural limitations from its
discrete-allocation formulation. It is undefined when no allocation
satisfies both agents' reservation utilities. It conditions only on
declared values, ignoring the preference signals embedded in negotiation
dialogue~\citep{Lewis2017, Chawla2021}. Its exhaustive search over
feasible allocations also grows exponentially with the number of
issues~\citep{Chevaleyre2006IssuesMARA}.

Recent approaches address subsets of these limitations. Reinforcement
learning agents~\citep{Lewis2017, Brenting2024} incorporate dialogue
signals but optimize unilateral rewards and cannot guarantee bilateral
individual rationality. Large language model
agents~\citep{Bianchi2024NegotiationArena, kwon2025astra} exploit
conversational context but treat normative constraints as prompt-level
rules, providing no formal guarantees. Diffusion-based
planners~\citep{Janner2022Diffuser, Chi2023DiffusionPolicy} operate in
continuous spaces and support constraint integration but address
single-agent processes and do not enforce bilateral criteria. No prior
system simultaneously guarantees individual rationality, security
proximity, and equitability on real-world dialogue corpora.

This study introduces the CGD framework
to address these limitations within a unified generative architecture.
CGD encodes bilateral preference structures using a graph attention
encoder that captures the relational comparison between agent
preferences. This representation is fused with a dialogue encoder
through cross-attention. A denoising diffusion probabilistic
model~\citep{Ho2020} then generates recommendations in the continuous
bilateral utility space. An analytically derived inference-time guidance
mechanism~\citep{Ho2022CFG} enforces individual rationality, security
proximity, and equitability by injecting constraint-derived gradients
into each reverse diffusion step. The specific architectural choices,
including the selection of a dynamic graph attention encoder over static
alternatives, are motivated in Section~\ref{subsec:lit_gnn} and
empirically validated in Section~\ref{sec:eval}.

The framework provides three contributions aligned with the identified
structural limitations.

\begin{enumerate}[label=\textbf{(C\arabic*)}, leftmargin=3.5em,
itemsep=3pt]

\item \textbf{IR-Robust Generation via Continuous Relaxation.} CGD operates in a continuous utility space where individually rational utility vectors always exist, eliminating the hard solver failures typical of discrete allocation spaces. At inference time, an analytically derived guidance gradient guarantees a strictly positive utility increment for any IR-violating agent at each reverse diffusion step, provided the guidance scale exceeds a finite threshold. Whether these per-step corrections accumulate to full IR satisfaction depends on corpus distribution and calibrated guidance scale; empirical evaluation confirms near-perfect individual rationality across all evaluated corpora under standard calibration.

\item \textbf{Inference-time welfare recovery without retraining.}
A welfare guidance term constrained to preserve IR dominance recovers
Pareto efficiency at inference time while strictly maintaining
individual rationality, without modifying pre-trained model weights.
Normative compliance and welfare maximization are thereby decoupled
into separate operational phases within a single inference pipeline.

\item \textbf{Complexity-bounded inference independent of the
allocation space.}
The denoising inference cost scales with the model size and number of
diffusion steps, independently of the number of possible discrete
allocations. Exhaustive solver enumeration time grows exponentially
with the number of issues while CGD denoising time remains constant,
confirmed empirically across a range of issue counts and corpora.

\end{enumerate}

CGD is evaluated on three datasets spanning synthetic and human
negotiation domains: a synthetic corpus based on a symmetric Dirichlet
distribution~\citep{Bishop2006PRML}, the CaSiNo corpus of human
campsite negotiation dialogues~\citep{Chawla2021}, and the Deal or No
Deal corpus of integer-valued negotiation
instances~\citep{Lewis2017}. The results demonstrate that CGD achieves
near-zero security gaps and satisfies equitability across all corpora,
addressing normative criteria that the NBS structurally cannot enforce.

The remainder of this paper is organized as follows.
Section~\ref{sec:litreview} reviews related work and formalizes the
research gap. Section~\ref{sec:problem} defines the negotiation model
and normative criteria. Section~\ref{sec:method} presents the CGD
architecture and its theoretical properties. Section~\ref{sec:setup}
describes the datasets, baselines, and evaluation protocols.
Section~\ref{sec:eval} reports the empirical results.
Section~\ref{sec:discussion} presents the implications of the study.
Section~\ref{sec:conclusion} concludes the paper.

\section{Literature Review}
\label{sec:litreview}

The CGD framework spans four intersecting research streams:
computational negotiation and NSSs,
dialogue-grounded strategic reasoning, graph neural networks for
multi-agent interaction, and diffusion-based generative models for
constrained decision-making.

\subsection{Automated Negotiation and Negotiation Support}
\label{subsec:lit_negotiation}

Bilateral multi-issue negotiation has been studied as a computational
problem for several decades. The field has shifted from handcrafted
bidding strategies derived from time-dependent concession
rules~\citep{Baarslag2017} toward end-to-end learnable policies.
Bagga et al.~\citep{Bagga2022} introduce deep learnable strategy
templates trained on utility-labeled traces that transfer across
preference profiles without manual redesign. Renting
et al.~\citep{Brenting2024} propose a graph neural network policy that
encodes problem structure as a graph and learns via reinforcement
learning, matching canonical baselines without domain-specific feature
engineering. Zhang et al.~\citep{Zhang2026} recast negotiation as a
market-design problem and show that deep reinforcement learning agents
converge to high-welfare, near-symmetric equilibria. These systems
demonstrate that learned policies can match or exceed analytical methods
in welfare. However, all three optimize individual rewards and cannot
guarantee bilateral individual rationality or security compliance.

The dominant analytical benchmark for negotiation support remains the
NBS~\citep{Nash1950}. Nakamura~\citep{Nakamura2024} provides
efficiency-free characterizations showing that individual rationality
can replace Pareto optimality in Nash's theorem.
Zeng et al.~\citep{NeurIPS2024Bargaining} embed the NBS in a
meta-learning framework for fairness-aware optimization, demonstrating
that Nash bargaining extends naturally to deep learning pipelines.
Neither result addresses the NBS's structural absence of security
proximity or equitability guarantees. Empirical evaluation confirms
that satisfying individual rationality does not imply satisfying
security proximity (Section~\ref{sec:eval}), motivating security
proximity as an independent normative criterion.

Wu and Sun~\citep{DeepQ2025} propose an agent-based model integrating
a fairness function and Q-learning that reduces utility disparity
relative to classical strategies. Their scalar fairness term conflates
individual rationality with distributional symmetry and provides formal
guarantees for neither. CGD addresses this by decomposing normative
compliance into three analytically distinct constraints enforced through
separate terms in the guidance gradient.

\subsection{Dialogue-Grounded Negotiation}
\label{subsec:lit_dialogue}

Lewis et al.~\citep{Lewis2017} establish dialogue-grounded negotiation
as a distinct task through the Deal or No Deal corpus.
Chawla et al.~\citep{Chawla2021} extend this setting with the CaSiNo
corpus. In both corpora, negotiation unfolds in natural language,
requiring models to infer preferences from conversational context.
Large language models currently dominate this task.
Kwon et al.~\citep{kwon2025astra} propose an LLM-based agent combining
linear-programming offer optimization, structured opponent modeling, and
Tit-for-Tat reciprocity, achieving competitive individual utility at
the cost of a nontrivial walk-away rate.
Bianchi et al.~\citep{Bianchi2024NegotiationArena} show that
GPT-4-class models secure mutually beneficial deals but fail under
shifts in stakes or tactics.
Chawla et al.~\citep{Chawla2023Selfish} show that selfish
reinforcement learning does not inevitably degrade fairness.
Bhattacharya et al.~\citep{Bhattacharya2025} find that current large
language models still struggle with commitment and legitimacy in
high-stakes settings.

A persistent limitation across these systems is that LLM-based agents
optimize individual objectives and cannot formally guarantee bilateral
constraints such as individual rationality or security. The empirical
symmetry recovered by Chawla et al.~\citep{Chawla2023Selfish} is not
formally guaranteed and does not address security proximity. CGD
differs from these approaches by enforcing normative constraints via a
deterministic, analytically derived guidance gradient rather than
through prompt engineering or emergent training behaviour.

\subsection{Graph Neural Networks for Strategic Reasoning}
\label{subsec:lit_gnn}

Negotiation is inherently relational. Agents correspond to nodes,
issues correspond to attributes, and the feasibility of an allocation
depends on the joint preference graph. Graph neural networks therefore
provide a natural inductive bias for encoding negotiation structure.
Renting et al.~\citep{Brenting2024} demonstrate that graph
representations generalize across preference profiles where flat-vector
encoders fail. Fu et al.~\citep{Fu2024GNN} show that graph-structured
state representations in actor-critic models can elicit cooperation and
adaptation simultaneously, while flat-vector baselines collapse to
greedy strategies under opponent diversity.

The original Graph Attention Network (GAT)~\citep{Velickovic2018} uses static
attention coefficients computed as a function of each node's features
independently. Brody et al.~\citep{Brody2022} identify a key limitation
of this design: static attention cannot represent the class of attention
functions that depend on a joint comparison of both endpoint features.
Their Graph Attention Network v2 (GATv2) architecture addresses this by computing attention
coefficients from a joint nonlinear function of both source and
destination features. In bilateral negotiation, the signal that
determines whether a mutually beneficial trade exists is comparative.
Agent 1's value for item $j$ relative to agent 2's value for item $j$
determines who should receive it. Static attention evaluates each
agent's features in isolation and cannot directly represent this
comparative signal. Dynamic attention in GATv2 computes coefficients
from the joint features of both agents on each edge, directly encoding
the bilateral comparison. This architectural argument motivates the use
of GATv2 as the strategic encoder in CGD.

Cross-attention between graph and sequence representations provides the
fusion mechanism between the strategic and dialogue modalities.
Faber and Wattenhofer~\citep{CommFormer2024} show that a shared
Transformer actor with a graph-based centralized critic captures
long-range dependencies in networked multi-agent control, which is
architecturally similar to CGD's fusion of the graph strategic encoder
with the dialogue encoder. CGD extends these ideas by conditioning a
continuous diffusion denoiser on the fused representation. No prior
work combines graph-structured bilateral preference encoding with
continuous diffusion generation and inference-time normative guidance
for negotiation support.

\subsection{Diffusion Models for Decision-Making and Planning}
\label{subsec:lit_diffusion}

Denoising diffusion probabilistic models~\citep{Ho2020} have been
extended from image generation to sequential decision-making and
planning. Janner et al.~\citep{Janner2022Diffuser} introduce Diffuser,
the first diffusion-based trajectory planner for offline reinforcement
learning, demonstrating that diffusion models can represent multimodal
outcome distributions and be steered at inference time through
differentiable cost functions. Chi et al.~\citep{Chi2023DiffusionPolicy}
extend this to visuomotor control, showing that diffusion-based action
generation outperforms Gaussian and mixture-of-Gaussian policies on
manipulation benchmarks. Li~\citep{li2023efficient} proposes planning
in a learned latent action space, supporting an equivalence between
latent planning and energy-guided sampling that underpins CGD's
continuous utility relaxation. A comprehensive survey by
Zhu et al.~\citep{DiffusionRLSurvey2024} identifies constrained
generation as the central open problem: existing methods can encourage
constraint satisfaction through residual loss terms but cannot
guarantee it.

Classifier-free guidance~\citep{Ho2022CFG} circumvents retraining
costs by steering the denoising process directly at inference time.
Zhang et al.~\citep{Romer2025Constraints} recently applied this
principle to robot control, demonstrating that guiding pre-trained
diffusion policies can enforce novel geometric constraints during
deployment. Ajay et al.~\citep{Ajay2023DecisionDiffuser} and
Chen et al.~\citep{SimpleHD2024} extended diffusion planning to
hierarchical settings via subgoal conditioning. These frameworks are
confined to single-agent environments and lack mechanisms for bilateral
constraint enforcement.

CGD adapts classifier-free guidance to multi-agent normative settings.
Rather than relying on a learned classifier, the guidance gradient in
CGD is derived analytically from individual rationality, security, and
equitability penalty functions. This analytical derivation yields
provably monotone per-step corrections, as established in
Section~\ref{subsubsec:l1_theory}.

\subsection{Fairness, Normative Constraints, and AI}
\label{subsec:lit_fairness}

Bhattacharya et al.~\citep{Bhattacharya2025} show that fairness in
large language model bargaining emerges from pre-training rather than
from an explicit guarantee. Mouri Zadeh Khaki and
Choi~\citep{Hua2024MDPI} show that fairness violations correlate with
aggressive reward-maximizing prompting and that third-party norm
enforcement partially mitigates them without providing formal
guarantees. Wang~\citep{GAME2025MultiObjective} establishes the
existence of Pareto-Nash equilibria in finite multi-objective Markov
games, providing theoretical support for treating individual
rationality, security proximity, and equitability as simultaneously
constrained objectives.

Lewis et al.~\citep{Lewis2017} provide clear evidence that
reward-maximizing fine-tuning destroys bilateral symmetry in Deal or
No Deal. Chawla et al.~\citep{Chawla2023Selfish} show that curriculum
design can recover near-symmetric outcomes.
Zhang et al.~\citep{Zhang2026} show that decentralized agents converge
to equitable splits when market mechanisms select for them. Together,
these results indicate that normative compliance cannot be reliably
achieved through training signal design alone. An architectural
mechanism that enforces constraints at inference time, independently of
the training objective, is therefore necessary.

\subsection{Research Gap and Positioning of CGD}
\label{subsec:lit_gap}

\Cref{tab:related} summarizes the normative coverage of representative
recent systems. No prior work simultaneously guarantees individual
rationality, security proximity, and equitability on real-world dialogue
corpora.

\begin{table}[h!]
\centering
\caption{Normative coverage of representative negotiation systems.
\ding{51}: guaranteed or formally verified. $\approx$: empirically
recovered but not guaranteed. $\times$: not addressed or violated by
design.}
\label{tab:related}
\footnotesize
\setlength{\tabcolsep}{4pt}
\renewcommand{\arraystretch}{1.2}
\begin{tabular}{lcccc}
\toprule
\textbf{System} & \textbf{IR} & \textbf{Security}
& \textbf{Symmetry} & \textbf{Pareto} \\
\midrule
NBS \citep{Nash1950}
  & $\approx$ & $\times$ & $\times$ & \ding{51} \\
CGO (this work, ablation)
  & $\approx$ & $\times$ & $\times$ & $\times$ \\
RNN-SL \citep{Lewis2017}
  & $\times$ & $\times$ & $\approx$ & $\times$ \\
RNN-RL \citep{Lewis2017}
  & $\times$ & $\times$ & $\times$ & $\times$ \\
Selfish-RL \citep{Chawla2023Selfish}
  & $\times$ & $\times$ & $\approx$ & $\times$ \\
GNN-RL \citep{Brenting2024}
  & $\times$ & $\times$ & $\times$ & $\times$ \\
ASTRA \citep{kwon2025astra}
  & $\approx$ & $\times$ & $\times$ & $\times$ \\
\textbf{CGD (this work)}
  & \ding{51} & \ding{51} & \ding{51} & $\approx$ \\
\bottomrule
\end{tabular}
\end{table}

Three specific gaps motivate CGD. First, agent-based
systems~\citep{Lewis2017, Brenting2024} optimize individual utility and
cannot guarantee bilateral individual rationality. Second, LLM-based
systems~\citep{kwon2025astra, Bianchi2024NegotiationArena} exploit
linguistic context but treat normative compliance as a
prompt-engineering problem rather than an architectural guarantee.
Third, diffusion planners for
decision-making~\citep{Janner2022Diffuser, Chi2023DiffusionPolicy}
operate in single-agent environments with no mechanism for bilateral
constraint enforcement.

The contribution of CGD lies in a specific architectural composition designed to close all three gaps simultaneously. Specifically, a graph attention encoder captures the bilateral relational structure of preferences, and a dialogue encoder provides complementary preference signals from the conversational context. Once cross-attention fuses these modalities, a continuous diffusion denoiser generates recommendations in utility space. Finally, an analytically derived guidance gradient enforces individual rationality, security proximity, and equitability at inference time without retraining. While each component can be found in prior work in other domains, the overarching contribution is their composition for the bilateral normative negotiation support problem, together with the analytical guidance gradient derived from cooperative bargaining criteria. The formal problem statement follows in Section~\ref{sec:problem}.

\section{Problem Statement}
\label{sec:problem}

Bilateral negotiations require two self-interested parties to jointly
allocate a set of resources under incomplete information. A NSS acts as a third-party recommender that proposes
outcomes satisfying individual rationality, procedural security, equity,
and Pareto efficiency. No single allocation simultaneously optimizes all
four criteria, and any method that searches over discrete allocations or
conditions only on declared utilities inherits at least one structural
limitation. This section formalizes the environment, defines each
normative criterion, identifies the tensions among them, and specifies
the three structural limitations that motivate CGD.

\subsection{Negotiation Environment}
\label{subsec:environment}

Let $\Nagents = \{1, 2\}$ be the set of agents and
$\Items = \{1, \ldots, J\}$ the set of indivisible issues, each
assigned to exactly one agent. Agent $i \in \Nagents$ holds a private
value vector $\bv{v}_i \in \R^{J}_{\geq 0}$ normalized so that
$\sum_{j} v_{ij} = 1$. The joint preference structure forms the value
matrix $\bv{V} \in \R^{|\Nagents| \times |\Items|}$ with rows
$\bv{v}_i$, which is never observed directly and must be inferred from
self-reports and dialogue. The information asymmetry is observational
rather than distributional: the framework learns to extract preference
signals from observable behaviour rather than modelling uncertainty over
agent type distributions.

\begin{definition}[Allocation]
\label{def:allocation}
An allocation is a binary matrix
$X \in \{0,1\}^{|\Nagents| \times |\Items|}$ satisfying
\begin{equation}
  \sum_{i \in \Nagents} X_{ij} = 1 \qquad \forall\, j \in \Items.
  \label{eq:feasibility}
\end{equation}
The set of all feasible allocations is denoted $\Alloc$.
\end{definition}

\begin{definition}[Agent utility]
\label{def:utility}
The utility of agent $i$ under allocation $X$ is
\begin{equation}
  u_i(X) = \sum_{j \in \Items} X_{ij}\,v_{ij} \in [0,1].
  \label{eq:utility}
\end{equation}
\end{definition}

The agents communicate through a negotiation dialogue $D$ comprising
natural language utterances exchanged prior to agreement. These
utterances carry implicit signals about priorities, concession
willingness, and reservation utilities that supplement the declared
value matrix $\bv{V}$.

\subsection{Normative Criteria}
\label{subsec:criteria}

Four normative criteria are defined below. Individual rationality and
Pareto efficiency are standard in cooperative bargaining theory.
Security proximity and equitability are operational additions specific
to the discrete multi-issue setting.

\begin{definition}[Individual Rationality]
\label{def:ir}
An allocation $X$ is individually rational if
\begin{equation}
  u_i(X) > b_i \quad \forall\, i \in \Nagents,
  \label{eq:ir}
\end{equation}
where $b_i \in [0,1]$ is agent $i$'s reservation utility, the payoff
from its best outside option. The IR rate of a recommendation policy is
$\mathrm{IR} = \mathbb{E}_X[\prod_i \mathbf{1}[u_i(X) > b_i]]$.
\end{definition}

\begin{remark}[Continuous versus discrete IR]
\label{rem:ir_continuous}
CGD generates recommendations in the continuous utility space
$\R^{|\Nagents|}$, where IR-satisfying vectors always exist. The scalar
utility $\hat{u}_i$ of the continuous recommendation substitutes for
$u_i(X)$ during guidance and evaluation. The projection step mapping
the continuous recommendation to the nearest feasible discrete
allocation may reduce utility and produce a discrete allocation that
violates IR even when the continuous recommendation does not. Under
consistent zero Best Alternative to a Negotiated Agreement (BATNA) evaluation conditions this projection gap is
empirically negligible~(Section~\ref{sec:eval}). The operational advantage of the continuous relaxation is
that CGD always produces a recommendation, whereas a discrete solver
that finds $\AllocIR = \emptyset$ produces no output.
\end{remark}

\begin{definition}[Security Level and Proximity]
\label{def:security}
Agent $i$'s security level is the utility obtained by retaining the
$\lceil J/2\rceil$ most-valued issues:
\begin{equation}
  \sigma_i = \sum_{j \in T_i} v_{ij}, \quad
  T_i = \mathrm{top}_{\lceil J/2\rceil}(\bv{v}_i).
  \label{eq:security}
\end{equation}
An allocation is security-proximate if
\begin{equation}
  \Delta_\sigma(X) = \sum_{i \in \Nagents}
  \max(\sigma_i - u_i(X),\, 0) \leq \varepsilon
  \label{eq:sec_gap}
\end{equation}
for tolerance $\varepsilon \geq 0$.
\end{definition}

\begin{remark}[Relation to the Nash Bargaining Solution]
\label{rem:nbs}
The classical NBS~\citep{Nash1950} maximizes the Nash product $\prod_i (u_i - d_i)$ 
over a convex feasible set given an exogenous disagreement point $d$. 
Definition~\ref{def:security} departs from this paradigm in two respects. 
First, the security level $\sigma_i$ is endogenous: rather than an exogenous payoff 
or a minimax value of the discrete game, it is a proportional-share reference 
derived directly from the value matrix $\bv{V}$. The $\lceil J/2\rceil$ norm establishes 
a procedural fairness standard rather than a strict strategic guarantee. Second, 
because $\sigma_i$ is computable entirely from self-reported values, it imposes no 
additional elicitation burden. In settings where the true disagreement point $d_i$ 
is known, it naturally supersedes both $b_i$ in Definition~\ref{def:ir} and $\sigma_i$ 
in Definition~\ref{def:security}. Crucially, because the NBS optimizes subject 
only to individual rationality, it provides no structural mechanism to enforce 
security proximity. Empirical evaluations (Section~\ref{sec:eval}) confirm this 
limitation, justifying security proximity as an analytically distinct normative 
requirement.
\end{remark}

\begin{definition}[Equitability]
\label{def:equitability}
An allocation $X$ is equitable if
\begin{equation}
    \Delta_{\mathrm{eq}}(X) = \bigl|u_1(X) - u_2(X)\bigr| \leq \delta
    \label{eq:equitability}
\end{equation}
for tolerance $\delta \geq 0$.
\end{definition}

\begin{remark}[Relation to Nash's symmetry axiom]
\label{rem:nash_symmetry}
Nash's symmetry axiom~\citep{Nash1950} states that if the feasible
utility set is symmetric about the diagonal and disagreement payoffs
are equal, then the Nash solution assigns equal utilities. This is a
conditional property of the solution concept, not a constraint on
individual allocations. Definition~\ref{def:equitability} is stronger:
it requires bounded inequality for any allocation regardless of whether
the feasible set is symmetric. The NSS has a duty to both parties and
cannot condition equitability on structural properties that may not hold
in practice.
\end{remark}

\begin{remark}[Normalized equitability metric and threshold]
\label{rem:sym_normalisation}
Definition~\ref{def:equitability} bounds the raw difference
$|u_1 - u_2|$, which is not commensurable across negotiations with
different joint welfare levels. The normalized form
$|u_1 - u_2|\,/\,(u_1 + u_2 + \varepsilon)$ is scale-invariant and
aligns with the Nash axiom of invariance to positive affine utility
transformations~\citep{Nash1950}. It is applied uniformly to all
methods throughout this paper. A threshold of $\Delta_{\mathrm{eq}}
\leq 0.15$ is used operationally. This value is not derived from a
game-theoretic axiom. It represents a practitioner tolerance requiring
that no agent receives less than approximately 43\% of joint welfare,
and is substantially stricter than the random allocation baseline
($\Delta_{\mathrm{eq}} \approx 0.55$) and the greedy efficient
allocation ($\Delta_{\mathrm{eq}} \approx 0.43$). The threshold is
the tightest round value under which \textsc{CGD} satisfies
equitability across all evaluated corpora while \textsc{NBS} violates
it on two of three, making it the smallest value that meaningfully
discriminates between the two systems; sensitivity to this choice is
verifiable from the symmetry gap columns of
Table~\ref{tab:sota}.\footnote{Reading
Table~\ref{tab:sota} at three thresholds: at $\delta =
0.10$, \textsc{CGD} passes on all three corpora (symmetry gaps
$0.099$, $0.021$, $0.006$) while \textsc{NBS} fails on all three
($0.182$, $0.283$, $0.136$); at $\delta = 0.15$, \textsc{CGD}
continues to pass on all three while \textsc{NBS} passes only on
\textsc{DND} ($0.136$); at $\delta = 0.20$ the discrimination pattern
is unchanged. \textsc{CGD} compliance holds across all three
thresholds, confirming the result is not an artefact of the chosen
value.} The threshold is adjustable at inference time through
$w_{\mathrm{eq}}$ without retraining.
\end{remark}

\begin{definition}[Pareto Efficiency]
\label{def:pareto}
An allocation $X$ is Pareto-efficient if no alternative $X' \in \Alloc$
improves one agent's utility without reducing the other's. The Pareto
gap is
\begin{equation}
  \Delta_{\mathrm{P}}(X) =
  \frac{W^* - (u_1(X)+u_2(X))}{W^*} \in [0,1],
  \label{eq:pareto}
\end{equation}
where $W^* = \max_{X' \in \Alloc}\{u_1(X')+u_2(X')\}$.
\end{definition}

\subsection{Criterion Tensions and the Support Problem}
\label{subsec:tensions}

The four criteria are mutually inconsistent in general. Three tensions
arise routinely in the multi-issue setting.

\textbf{Tension 1: IR versus security-proximity.} When both agents hold
high reservation utilities $b_i$ and high security levels $\sigma_i$,
the IR-feasible set
$\AllocIR = \{X \in \Alloc : u_i(X) > b_i\ \forall\, i\}$
may be empty:
\begin{equation}
  \max_{X \in \Alloc}\,\min_{i \in \Nagents}
  \bigl(u_i(X) - b_i\bigr) \leq 0.
  \label{eq:infeasibility}
\end{equation}
The NBS is undefined in this case. The practitioner faces a negotiation
in which no allocation is acceptable to both parties, precisely when
support is most needed.

\textbf{Tension 2: Pareto efficiency versus equitability.} The
allocation maximizing joint surplus assigns each issue to the agent that
values it most. When one agent's values dominate across all issues, this
allocation gives the dominant agent near-total utility and the other
near-zero utility. Equitability therefore conflicts with Pareto
efficiency whenever agents have heterogeneous value profiles, which is
the generic case.

\textbf{Tension 3: Security-proximity versus equitability.} When the
top-$\lceil J/2 \rceil$ item sets $T_1$ and $T_2$ overlap
substantially, assigning contested items to one agent necessarily
reduces the other below its security level. Even when a jointly
security-proximate allocation exists, it may be asymmetric when
$\sigma_1 \neq \sigma_2$.

Relaxing any one criterion has consequences for the parties involved.
This motivates a formulation that imposes all four constraints
simultaneously with practitioner-adjustable tolerances.

\begin{definition}[Negotiation Decision Support Problem]
\label{def:ndsp}
Given $\bv{V}$, $\bv{b} = (b_1, b_2)$,
$\bv{\sigma} = (\sigma_1, \sigma_2)$,
tolerances $(\varepsilon, \delta, \eta) \geq 0$, and dialogue $D$,
the NDSP is to find $X^*$ satisfying:
\begin{alignat}{2}
  &\text{(IR)} \quad  & u_i(X^*)                  &> b_i
    \quad\forall\,i, \label{eq:c_ir} \\
  &\text{(Sec)} \quad & \Delta_\sigma(X^*)         &\leq \varepsilon,
    \label{eq:c_sec} \\
  &\text{(Eq)} \quad  & \Delta_{\mathrm{eq}}(X^*)  &\leq \delta,
    \label{eq:c_eq} \\
  &\text{(Par)} \quad & \Delta_{\mathrm{P}}(X^*)   &\leq \eta.
    \label{eq:c_par}
\end{alignat}
\end{definition}

\subsection{Structural Limitations}
\label{subsec:limitations}

Three structural limitations prevent the NDSP from being solved by
existing approaches in general cases.

\textbf{L1: Discrete infeasibility.} Any framework searching over
$\Alloc$ fails to produce a recommendation when
$\AllocIR = \emptyset$ under condition~\eqref{eq:infeasibility}. A
framework operating in the continuous utility space $\R^{|\Nagents|}$
is not subject to this constraint, as IR-satisfying vectors always
exist in the continuous relaxation. The nearest feasible discrete
allocation is recovered in post-processing. The operational advantage
is that CGD always produces a recommendation, whereas a discrete solver
fails entirely when $\AllocIR = \emptyset$. If the discrete allocation
space genuinely contains no IR-feasible allocation, the projection step
cannot recover one. The continuous guidance minimizes the residual IR
violation before projection. The magnitude of this residual gap is
quantified empirically in Section~\ref{sec:eval}.

\textbf{L2: Preference elicitation gap.} Optimization-based NSSs
condition on the declared value matrix $\bv{V}$, which is an imperfect
proxy for true preferences. The negotiation dialogue $D$ contains
complementary signals about issue salience and concession limits absent
from $\bv{V}$. Realising this capacity requires two separable
conditions to hold: the corpus must carry recoverable preference signals
beyond $\bv{V}$, and the encoder must have sufficient semantic
competence to extract them. Both conditions are tested empirically in
Section~\ref{sec:eval} through a three-way architectural ablation and a
controlled misrepresentation experiment.

\textbf{L3: Exponential search complexity.} Exhaustive evaluation of
all $|\Alloc| = |\Nagents|^{|\Items|}$ allocations grows exponentially
in $|\Items|$. Two computational steps must be distinguished. The
denoising inference step operates in $\R^{|\Nagents|}$ and can be made
independent of $|\Alloc|$, scaling only with the model size and number
of diffusion steps. The post-processing step, which maps the continuous
recommendation to the nearest feasible discrete allocation, requires a
search over $\Alloc$ and retains the exponential dependence on
$|\Items|$. The continuous relaxation therefore shifts the exponential
bottleneck from the optimization phase to the post-processing phase
rather than eliminating it. Empirical timing results in
Section~\ref{sec:eval} quantify both components across
$|\Items| \in \{3, 5, 7, 10\}$, confirming constant denoising time
and exponentially growing NBS enumeration time across all evaluated
corpora. The crossover point at which NBS enumeration time exceeds CGD
denoising time is estimated at approximately $|\Items| \approx 17$.

\subsection{Research Objective}
\label{subsec:objective}

Limitations L1 through L3 are structural properties of the discrete
allocation formulation rather than deficiencies of particular
algorithms. Addressing all three simultaneously requires generating
recommendations in continuous utility space, conditioning jointly on
structured preference data and dialogue, and producing recommendations
at inference cost independent of $|\Alloc|$.

The CGD framework is designed to satisfy these requirements through
four operational objectives.

\begin{enumerate}[label=\textbf{(O\arabic*)}, leftmargin=3.5em,
                  itemsep=3pt]
  \sloppy

  \item \textbf{IR-completeness.}
    Recommendations satisfy individual rationality by operating in the
    continuous utility relaxation and applying iterative normative
    corrections during generation, circumventing the discrete
    infeasibility condition (Equation~\eqref{eq:infeasibility}). Discrete IR
    compliance under consistent evaluation conditions is empirically
    confirmed across all corpora.

  \item \textbf{Contextual grounding.}
    Recommendations are conditioned on graph-structured strategic
    context and negotiation dialogue, incorporating both declared and
    dialogue-revealed preference signals where recoverable. The
    conditions under which dialogue signals are recoverable are
    identified empirically in Section~\ref{sec:eval}.

  \item \textbf{Welfare recoverability.}
    The inference pipeline supports optional welfare guidance that
    recovers Pareto-efficient discrete allocations without model
    retraining, while preserving all normative constraints enforced by
    the primary guidance gradient.

  \item \textbf{Complexity-bounded inference.}
    The denoising inference cost is $O(T \cdot |\theta|)$, independent
    of $|\Alloc| = |\Nagents|^{|\Items|}$. This architectural
    decoupling is empirically confirmed across multiple issue counts
    and all evaluated corpora.

\end{enumerate}

Whether and to what degree these properties are achieved is established
empirically in Section~\ref{sec:eval}.

\section{Methodology}
\label{sec:method}

The three limitations identified in Section~\ref{sec:problem} stem from
two structural constraints: negotiation outcomes are represented as
discrete allocations $X \in \Alloc$, and recommendations are conditioned
solely on the declared value matrix $\bv{V}$. Limitation~L1 emerges
because the finite set $\Alloc$ may contain no IR-feasible allocation.
Limitation~L2 arises because $\bv{V}$ is an incomplete proxy for true
preferences when the dialogue $D$ carries implicit preference signals.
Limitation~L3 restricts scalability because any algorithm enumerating
$\Alloc$ incurs costs that grow exponentially in the number of issues.

Addressing all three requires generating recommendations in the
continuous utility space $\R^{|\Nagents|}$, conditioning them jointly on
$\bv{V}$ and $D$, and ensuring that the denoising inference cost is
independent of $|\Alloc|$. The CGD framework is designed to satisfy
these requirements. Each architectural component is detailed below, and
Section~\ref{sec:theory} establishes the corresponding theoretical
properties. The operational scope of each theoretical result and its
relationship to the empirical findings are discussed in
Section~\ref{sec:discussion}.

\subsection{Architectural Overview}
\label{subsec:overview}

CGD processes two inputs to produce a single output. The first is the
strategic context $\mathcal{G} = (\bv{V},\bv{b},\bv{\sigma},\bv{r})$,
comprising the value matrix, reservation utilities, security levels, and
role indicators. The second is the dialogue context $D$, comprising the
natural language utterances exchanged by the negotiating parties.

These inputs are mapped to a joint utility tensor
$\hat{\bv{u}} \in \R^{|\Nagents| \times d_u}$. The denoiser operates in
a $d_u$-dimensional latent space per agent. The final scalar utility
recommendation for agent $i$ is extracted from the first component,
$\hat{u}_i \equiv \hat{u}_{i,1}$. The remaining $d_u - 1$ dimensions
serve as auxiliary latent capacity during generation and are discarded
prior to evaluation. All normative criteria and metrics are evaluated on
the clamped value $\hat{u}_i \in [0,1]$.

The transformation proceeds through four stages, illustrated in
Figure~\ref{fig:architecture}. A graph attention encoder maps
$\mathcal{G}$ into a graph-structured embedding $\bv{H}^{\mathrm{strat}}$
capturing relational dependencies between agents. In parallel, a
dialogue encoder processes $D$ into token-level representations
$\bv{H}^{\mathrm{txt}}$. A cross-attention layer fuses these embeddings
into a unified conditioning context $\bv{h}^{\mathrm{ctx}}$, which is
injected into every layer of the denoiser via feature-wise modulation.
A U-Net denoiser then iteratively transforms Gaussian noise $\bv{u}_T$
into the recommendation $\hat{\bv{u}}_0$ via DDPM reverse diffusion.
During this reverse process, normative guidance corrections align the
generated output with individual rationality, security proximity, and
equitability.

\begin{figure}[ht]
  \centering
  \resizebox{\textwidth}{!}{\begin{tikzpicture}[
    box/.style={draw=ejorblue, fill=lightblue, rounded corners=3pt,
      minimum width=2.5cm, minimum height=0.75cm,
      font=\small\bfseries, align=center, text=ejorblue, inner sep=4pt},
    rbox/.style={draw=warningred, fill=lightyellow, rounded corners=3pt,
      minimum width=3.2cm, minimum height=0.7cm,
      font=\small\bfseries, align=center, text=warningred,
      inner sep=4pt, dashed},
    gbox/.style={draw=ejorblue!70, fill=lightgreen, rounded corners=3pt,
      minimum width=3.0cm, minimum height=1.0cm,
      font=\small\bfseries, align=center,
      text=ejorblue!80!black, inner sep=4pt},
    arr/.style={-{Stealth[length=5pt]}, thick, color=ejorblue},
    rarr/.style={-{Stealth[length=5pt]}, thick, color=warningred,
                 dashed},
    lbl/.style={font=\scriptsize\itshape, color=darkgray},
    node distance=0.35cm and 0.9cm
  ]
    \node[box] (state) at (0,2.2)
      {Agent State\\{\footnotesize$\bv{V},\bv{b},\bv{\sigma},\bv{r}$}};
    \node[box, right=0.9cm of state] (gatv2)
      {GATv2\\Encoder};
    \node[box, right=0.9cm of gatv2] (hstrat)
      {$\bv{H}^{\mathrm{strat}}$};
    \node[box] (diag) at (0,0)
      {Dialogue\\Context $D$};
    \node[box, right=0.9cm of diag] (roberta)
      {Dialogue\\Encoder};
    \node[box, right=0.9cm of roberta] (htxt)
      {$\bv{H}^{\mathrm{txt}}$};
    \node[box, right=1.1cm of hstrat, yshift=-1.1cm] (ca)
      {Cross-Attention\\$\bv{h}^{\mathrm{ctx}}$};
    \node[gbox, right=1.3cm of ca] (unet)
      {U-Net Denoiser\\
       $\epsilon_\theta(\bv{u}_t,t,\bv{h}^{\mathrm{ctx}})$};
    \node[rbox, above=0.65cm of unet] (norm)
      {$\Lnorm^{\mathrm{ext}}$: IR $+$ Security\\
       $+$ Equitability $+$ Welfare};
    \node[box, right=1.1cm of unet] (out)
      {$\hat{\bv{u}}_0$};
    \node[lbl, left=0.55cm of unet] (noise) {Noise $\bv{u}_T$};

    \draw[arr] (state)   -- (gatv2);
    \draw[arr] (gatv2)   -- (hstrat);
    \draw[arr] (hstrat)  -- (ca);
    \draw[arr] (diag)    -- (roberta);
    \draw[arr] (roberta) -- (htxt);
    \draw[arr] (htxt)    -- (ca);
    \draw[arr] (ca)      -- (unet)
      node[midway,above,lbl]{condition};
    \draw[arr] (noise)   -- (unet);
    \draw[rarr](norm)    -- (unet)
      node[midway,right,lbl,text=warningred]
        {$-\gamma\nabla_{\!\bv{u}_t}\Lnorm^{\mathrm{ext}}$};
    \draw[arr] (unet)    -- (out)
      node[midway,above,lbl]{sample};

    \begin{scope}[on background layer]
      \node[draw=ejorblue!30, rounded corners=4pt, fill=softgray,
        fit=(state)(gatv2)(hstrat)(diag)(roberta)(htxt),
        label={[font=\footnotesize\bfseries,text=ejorblue]
               below:\strut Strategic \& Dialogue Encoding},
        inner sep=7pt] {};
      \node[draw=ejorblue!30, rounded corners=4pt, fill=softgray,
        fit=(ca)(unet)(out)(norm),
        label={[font=\footnotesize\bfseries,text=ejorblue]
               below:\strut Conditional Graph Diffusion},
        inner sep=7pt] {};
    \end{scope}
  \end{tikzpicture}}
  \caption{The CGD pipeline. \textit{Left}: The strategic context
  $(\bv{V},\bv{b},\bv{\sigma},\bv{r})$ and dialogue context $D$ are
  processed by a GATv2 encoder and a dialogue encoder, respectively.
  Cross-attention fuses these representations into a unified conditioning
  context $\bv{h}^{\mathrm{ctx}}$. \textit{Right}: The U-Net denoiser
  $\epsilon_\theta$ executes the DDPM reverse process conditioned on
  $\bv{h}^{\mathrm{ctx}}$. An extended normative guidance objective
  $\Lnorm^{\mathrm{ext}}$ iteratively corrects the noisy sample at each
  step (dashed arrow) without altering the denoiser weights $\theta$.}
  \label{fig:architecture}
\end{figure}
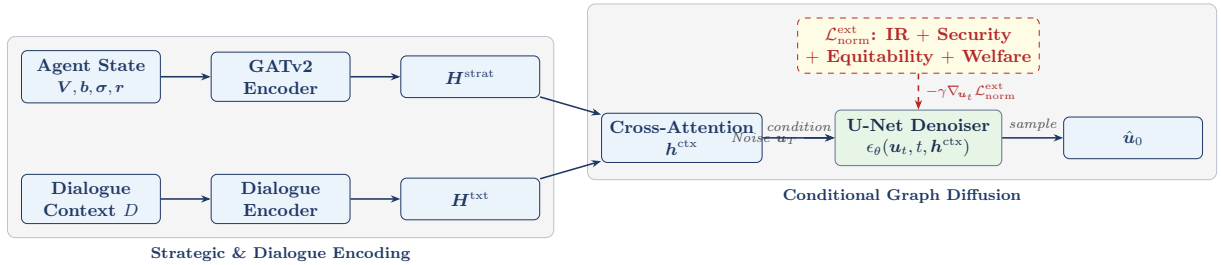

\subsection{Strategic Graph Construction via GATv2}
\label{subsec:gat}

The strategic encoder must represent not only each agent's individual
attributes but also the comparative relationship between them. In
bilateral negotiation, the signal that determines whether a mutually
beneficial trade exists is relational: agent $i$'s value for item $j$
relative to agent $k$'s value for the same item determines efficient
assignment. This relational structure motivates a graph representation.
The strategic context $\mathcal{G}$ is modeled as a directed graph
$G = (\mathcal{V},\mathcal{E})$ over agent nodes $\mathcal{V} = \Nagents$
connected by bidirectional edges. Each node $i$ is assigned the feature
vector
\begin{equation}
  \bv{s}_i =
  \bigl[\,b_i,\;\sigma_i,\;\bv{v}_i^\top,\;\bv{r}_i^\top,\;
         \bv{o}_i^\top,\;\bv{0}_{\mathrm{pad}}\,\bigr]
  \in \R^{F},
  \label{eq:node_feat}
\end{equation}
where $\bv{o}_i$ is an optional social-value-orientation encoding
(zero-padded when unavailable) and $\bv{0}_{\mathrm{pad}}$ ensures fixed
dimensionality $F$.

As established in Section~\ref{subsec:lit_gnn}, GATv2 is employed for the strategic graph encoding. For a directed edge $(k \to i)$ at attention head $h$, the unnormalized attention coefficient is computed as:
\begin{equation}
  e_{ki}^{(h)} =
  \bv{a}^{(h)\top}
  \operatorname{LeakyReLU}\!\bigl(
    \bv{W}_{\mathrm{src}}^{(h)}\bv{s}_k
    + \bv{W}_{\mathrm{dst}}^{(h)}\bv{s}_i
  \bigr).
  \label{eq:gatv2_logit}
\end{equation}
The attention weight that agent $i$ assigns to agent $k$ therefore
reflects the relative comparison of their parameters rather than
evaluating either agent's features in isolation. The final embedding is
obtained by softmax normalization and multi-head aggregation:
\begin{equation}
  \bv{h}_i =
  \operatorname{ELU}\!\Bigl(
    \bv{W}_{\mathrm{out}}
    \Bigl\|_{h=1}^{H}
    \sum_{k \in \mathcal{N}(i)}
    \alpha_{ki}^{(h)}\bv{W}_{\mathrm{dst}}^{(h)}\bv{s}_k
  \Bigr),
  \label{eq:gatv2_agg}
\end{equation}
yielding the strategic embedding matrix
$\bv{H}^{\mathrm{strat}} \in \R^{|\Nagents| \times d_{\mathrm{strat}}}$.
The choice of GATv2 over static alternatives is validated empirically in
Section~\ref{sec:eval} through an ablation comparing flat MLP, graph
convolutional network, original GAT, and GATv2 encoders on normative
compliance metrics.

\subsection{Dialogue Encoding}
\label{subsec:roberta}

Limitation~L2 arises because $\bv{V}$ is an incomplete proxy for true
preferences. The dialogue $D$ contains complementary signals about
concession intent, issue salience, and context-dependent shifts in
valuation. Whether these signals are extractable in practice depends on
two separable conditions: the dialogue must carry preference information
beyond $\bv{V}$, and the encoder must have sufficient semantic capacity
to extract it. Observation~\ref{obs:dialogue_necessity} in
Section~\ref{subsubsec:l2_theory} establishes the first condition as a
standard consequence of conditioning on informative variables. The second
is an empirical question; the ablation in
Section~\ref{subsec:dialogue_ablation} tests both conditions, and the
implications of the results are discussed in
Section~\ref{sec:discussion}.

A transformer-based dialogue encoder processes the sequence
$D = (w_1,\ldots,w_L)$ to produce token-level representations and a
pooled summary:
\begin{equation}
  \bv{H}^{\mathrm{txt}},\;\bv{h}^{\mathrm{txt}}
  = \mathrm{DialogueEnc}(D),
  \quad
  \bv{H}^{\mathrm{txt}} \in \R^{L \times d_{\mathrm{txt}}},\;
  \bv{h}^{\mathrm{txt}} \in \R^{d_{\mathrm{txt}}}.
  \label{eq:roberta_out}
\end{equation}
The token-level representations $\bv{H}^{\mathrm{txt}}$ preserve
fine-grained contextual dependencies for downstream cross-attention
fusion. The pooled summary $\bv{h}^{\mathrm{txt}}$ provides an auxiliary
conditioning signal.

\subsection{Cross-Attention Fusion}
\label{subsec:fusion}

The strategic embedding $\bv{h}^{\mathrm{strat}}$ and the dialogue token
sequence $\bv{H}^{\mathrm{txt}}$ operate at different granularities and
must be aligned before conditioning the denoiser. Cross-attention allows
the strategic summary to selectively attend to the most relevant dialogue
segments:
\begin{equation}
  \bv{h}^{\mathrm{ctx}} =
  \operatorname{softmax}\!\!\left(
    \frac{(\bv{W}_Q\bv{h}^{\mathrm{strat}})
          (\bv{W}_K\bv{H}^{\mathrm{txt}})^\top}{\sqrt{d_c}}
  \right)
  (\bv{W}_V\bv{H}^{\mathrm{txt}}),
  \label{eq:crossattn}
\end{equation}
where $\bv{W}_Q, \bv{W}_K, \bv{W}_V$ project inputs into a shared
latent space of dimension $d_c$. A residual connection to
$\bv{h}^{\mathrm{strat}}$ is applied after the attention mechanism.
The fused context $\bv{h}^{\mathrm{ctx}} \in \R^{d_c}$ is computed once
and broadcast to every denoiser layer via feature-wise modulation. This
ensures that conditioning cost remains constant across all $T$
reverse-diffusion steps.

\subsection{Denoising Diffusion Probabilistic Model (DDPM)}
\label{subsec:ddpm}

With $\bv{h}^{\mathrm{ctx}}$ encoding the negotiation context, the
objective is to generate a utility vector satisfying the NDSP criteria.
A diffusion model is suited to this task because it parameterizes a
flexible distribution over utility vectors, supports arbitrary
conditioning, and admits gradient-based corrections at inference time
without modifying the underlying model weights. 

\subsubsection{Forward Process}
\label{subsubsec:forward}

Let $\bv{u}_0 \in \R^{|\Nagents| \times d_u}$ be the target utility
vector drawn from the training distribution. The forward process corrupts
$\bv{u}_0$ over $T$ steps via a Gaussian Markov chain with closed-form
marginal:
\begin{equation}
  q(\bv{u}_t \mid \bv{u}_0) =
  \mathcal{N}\!\bigl(
    \bv{u}_t;\,
    \sqrt{\bar{\alpha}_t}\,\bv{u}_0,\,
    (1-\bar{\alpha}_t)\bv{I}
  \bigr),
  \label{eq:forward}
\end{equation}
where $\alpha_t = 1-\beta_t$,
$\bar{\alpha}_t = \prod_{s=1}^{t}\alpha_s$, and $\{\beta_t\}_{t=1}^{T}$
is a monotonically increasing noise schedule detailed in
Section~\ref{sec:implementation}.

\subsubsection{Denoiser: U-Net with FiLM Conditioning}
\label{subsubsec:denoiser}

The denoiser $\epsilon_\theta(\bv{u}_t,t,\bv{h}^{\mathrm{ctx}})$ is
implemented as a U-Net that estimates the noise added at step $t$,
conditioned on $\bv{h}^{\mathrm{ctx}}$ and the current time step $t$.
The encoder-decoder structure uses feedforward layers with skip
connections. A self-attention layer at the bottleneck across the agent
dimension allows each agent's representation to attend to the other,
capturing bilateral interactions. Skip connections preserve per-agent
feature asymmetry, accommodating the equity tensions identified in
Section~\ref{sec:problem}.

The conditioning signal is injected into every layer via Feature-wise
Linear Modulation (FiLM)~\citep{Perez2018}. Let
$\bv{C}^{(\ell)} = [\bv{h}^{\mathrm{ctx}}, \bv{e}_t]$ denote the
concatenation of the fused context and a sinusoidal time-step embedding
$\bv{e}_t$. A small MLP predicts affine parameters applied as:
\begin{equation}
  \widetilde{\bv{z}}_i^{(\ell)} =
  \bv{\gamma}^{(\ell)} \odot
  \operatorname{LayerNorm}(\bv{z}_i^{(\ell)})
  + \bv{\beta}^{(\ell)},
  \label{eq:film}
\end{equation}
where $(\bv{\gamma}^{(\ell)},\bv{\beta}^{(\ell)})$ is the MLP output.
Because $\bv{h}^{\mathrm{ctx}}$ remains constant across all $T$ reverse
steps, the FiLM parameters vary only with $t$, and inference cost scales
strictly as $O(T \cdot |\theta|)$.

\subsubsection{Training Objective}
\label{subsubsec:loss}

The model minimizes a composite loss combining DDPM reconstruction error
with a normative penalty:
\begin{equation}
  \mathcal{L}(\theta) =
  \underbrace{
    \E_{t,\bv{u}_0,\bv{\epsilon}}\!\bigl[
      \norm{\bv{\epsilon}
            - \epsilon_\theta(\bv{u}_t,t,\bv{h}^{\mathrm{ctx}})}^2
    \bigr]
  }_{\Ldiff}
  +
  \lambda\,\E_{\bv{u}_0}\!\bigl[\Lnorm(\bv{u}_0)\bigr],
  \label{eq:loss}
\end{equation}
where $\lambda > 0$ regulates the normative penalty $\Lnorm$
(Section~\ref{subsec:guidance}). This penalty regularizes the denoiser
during training, encouraging alignment with IR, security, and
equitability constraints prior to inference-time guidance. 
\subsection{Normative Guidance}
\label{subsec:guidance}

A trained diffusion model approximates the utility distribution of the
training corpus. Because real-world negotiation data often include
outcomes that violate IR, fall below security thresholds, or exhibit
significant utility asymmetry, normative guidance addresses these
discrepancies at inference time. Constraint-derived gradients are
injected into each reverse step, steering the generation process without
altering the pre-trained model weights.

\subsubsection{Normative Loss}
\label{subsubsec:norm_loss}

The normative loss formulates the three primary NDSP constraints as a
weighted sum of convex penalty terms:
\begin{equation}
  \Lnorm(\hat{\bv{u}}) =
  \underbrace{
    \sum_{i \in \Nagents} \relu{b_i - \hat{u}_{i,1}}
  }_{\mathcal{L}_{\mathrm{IR}}}
  +
  \underbrace{
    \sum_{i \in \Nagents} \relu{\sigma_i - \hat{u}_{i,1}}
  }_{\mathcal{L}_{\mathrm{sec}}}
  +
  \underbrace{
    w_{\mathrm{eq}}\,\bigl|\hat{u}_{1,1} - \hat{u}_{2,1}\bigr|
  }_{\mathcal{L}_{\mathrm{eq}}},
  \label{eq:norm_loss}
\end{equation}
where $w_{\mathrm{eq}} \geq 0$ controls the equitability penalty
strength. The three terms enforce Individual Rationality
(Definition~\ref{def:ir}), Security-Proximity
(Definition~\ref{def:security}), and Equitability
(Definition~\ref{def:equitability}) respectively. Pareto efficiency
is not explicitly penalized; it is pursued indirectly as the diffusion
model learns the distribution of successful negotiated outcomes. All
three terms are convex, non-negative, and subgradient-computable, and
$\Lnorm = 0$ is achieved if and only if both agents' utilities meet
their respective thresholds and are perfectly equal.

\subsubsection{Guided Reverse Process}
\label{subsubsec:reverse}

During inference, the standard DDPM update at each reverse step $t$ is
augmented by a gradient correction derived from an extended normative
loss $\Lnorm^{\mathrm{ext}}$:
\begin{equation}
  \bv{u}_{t-1} \leftarrow
  \underbrace{
    \frac{1}{\sqrt{\alpha_t}}\!\left(
      \bv{u}_t
      - \frac{1-\alpha_t}{\sqrt{1-\bar{\alpha}_t}}
        \epsilon_\theta(\bv{u}_t,t,\bv{h}^{\mathrm{ctx}})
    \right)
  }_{\bv{u}_{t-1}^{\mathrm{raw}}}
  -
  \gamma\,\nabla_{\bv{u}_t}\Lnorm^{\mathrm{ext}}(\bv{u}_t),
  \label{eq:guided_step}
\end{equation}
where $\gamma \geq 0$ is the guidance scale. This correction shifts the
sample trajectory while leaving model weights unchanged. The gradient is
computed by detaching $\bv{u}_t$ from the computational graph and
calling backward on $\Lnorm^{\mathrm{ext}}$ explicitly, ensuring that
the gradient step is not suppressed by the surrounding inference context.

The extended loss $\Lnorm^{\mathrm{ext}}$ augments
Equation~\eqref{eq:norm_loss} with a welfare guidance term weighted by
$w_{\mathrm{welfare}} \geq 0$. Its subgradient with respect to agent
$i$'s scalar utility is:
\begin{equation}
  \frac{\partial\Lnorm^{\mathrm{ext}}}{\partial \hat{u}_{i,1}}
  =
  -\ind{\hat{u}_{i,1} < b_i}
  -\ind{\hat{u}_{i,1} < \sigma_i}
  +w_{\mathrm{eq}}\operatorname{sign}(\hat{u}_{i,1}-\hat{u}_{j,1})
  -w_{\mathrm{welfare}}\bigl(1-\hat{u}_{i,1}\bigr)^{+},
  \label{eq:grad_lnorm}
\end{equation}
where $j \neq i$ and $(1-\hat{u}_{i,1})^{+}$ is agent $i$'s remaining
utility headroom. The welfare term pushes each agent's utility upward
proportionally to available headroom, preserving natural asymmetries
between agents at different utility levels. This separation of normative
and welfare objectives within the same inference pipeline constitutes
the operational basis of Contribution~C2: normative constraints are
enforced unconditionally by the primary gradient, and welfare is then
maximized subject to those constraints by the welfare term.

To guarantee that the IR correction remains dominant whenever an IR
violation occurs, the combined weights must satisfy:
\begin{equation}
  w_{\mathrm{eq}} + w_{\mathrm{welfare}} < 1.
  \label{eq:welfare_safety}
\end{equation}
When $\hat{u}_{i,1} < b_i$, both hinge indicators contribute $-1$
while the equitability and welfare terms contribute at most
$+w_{\mathrm{eq}}$ and $+w_{\mathrm{welfare}}$ respectively, yielding
$\partial\Lnorm^{\mathrm{ext}}/\partial\hat{u}_{i,1} \leq
-1 + w_{\mathrm{eq}} + w_{\mathrm{welfare}}$, which is strictly negative
if and only if Equation~\eqref{eq:welfare_safety} holds. This condition
bounds the valid hyperparameter space for $w_{\mathrm{welfare}}$ given
any chosen $w_{\mathrm{eq}}$. The theoretical consequence is established
in Section~\ref{subsubsec:l1_theory}.

\subsubsection{BATNA Estimation Strategies}
\label{subsubsec:batna}

The parameters $b_i$ and $\sigma_i$ in Equation~\eqref{eq:grad_lnorm}
are derived from $\bv{V}$ at inference time. Three strategies for
estimating the reservation utility $b_i$ are evaluated:
\begin{itemize}
  \item \textbf{Zero} ($b_i = 0$): Any positive utility satisfies IR,
    providing a conservative lower bound.
  \item \textbf{Floor} ($b_i \propto \max_j v_{ij}$): BATNA
    proportional to the agent's highest-valued issue.
  \item \textbf{Dataset} ($b_i = \hat{F}_i^{-1}(0.1)$): The 10th
    percentile of agent $i$'s realized utilities in the training corpus,
    requiring access to historical outcome data.
\end{itemize}
Security levels $\sigma_i$ are computed deterministically as the sum of
agent $i$'s top-$\lceil J/2 \rceil$ item values and require no
estimation. 

\subsection{Post-Processing: Recovery of Feasible Allocations}
\label{subsec:postproc}

The NDSP formally requires a discrete allocation $X^* \in \Alloc$. CGD
generates recommendations in $\R^{|\Nagents|}$ and recovers the final
allocation by mapping the continuous recommendation to the closest
discrete allocation in utility space:
\begin{equation}
  X^* = \argmin_{X \in \Alloc}\;
  \bigl\|\bv{u}(X) - \hat{\bv{u}}_0\bigr\|^2,
  \label{eq:postproc}
\end{equation}
where $\hat{\bv{u}}_0 = (\hat{u}_1,\hat{u}_2)$ is the clamped
continuous recommendation and $\bv{u}(X)$ is the utility vector
realized by allocation $X$. The search over $\Alloc$ is exhaustive.
Its cost grows exponentially in $|\Items|$ and is the computational
bottleneck at large item counts. Empirical timing results in
Section~\ref{sec:eval} confirm that post-processing time grows
exponentially across $|\Items| \in \{3,5,7,10\}$ while denoising time
remains constant.

A generalized weighted objective
$\argmin_{X \in \Alloc}\,\alpha\|\bv{u}(X)-\hat{\bv{u}}_0\|^2
-(1-\alpha)(u_1(X)+u_2(X))$
for $\alpha \in [0,1]$ interpolates between strict nearest-allocation
search ($\alpha=1$) and greedy maximum-welfare search ($\alpha=0$).
Empirical evaluation in Section~\ref{sec:welfare_results} confirms that
$\alpha^* = 1.0$ universally across all corpora.

Two welfare metrics are distinguished throughout. The continuous welfare
$W_{\mathrm{cont}} = \mathbb{E}[\hat{u}_1+\hat{u}_2]$ measures
recommendation quality prior to post-processing. The discrete welfare
$W_{\mathrm{disc}} = \mathbb{E}[u_1(X^*)+u_2(X^*)]$ measures actual
utility after allocation recovery. Because all baseline methods operate
exclusively in $\Alloc$, comparative evaluations use $W_{\mathrm{disc}}$
throughout.

\subsection{Theoretical Analysis}
\label{sec:theory}

The three structural limitations are addressed through different
analytical mechanisms. Limitation~L3 is circumvented by the
architectural design. Limitation~L1 is mitigated through a bounded
guidance scale. Limitation~L2 is addressed by expanding representational
capacity. Each is formalized below.

\subsubsection{L1: Discrete Infeasibility}
\label{subsubsec:l1_theory}

\begin{definition}[Per-step denoising magnitude]
\label{def:delta}
The per-step denoising contribution to agent $i$'s scalar utility at
reverse step $t$ is:
\begin{equation}
  \delta_t^{(i)} =
  \frac{1-\alpha_t}{\sqrt{\alpha_t(1-\bar{\alpha}_t)}}\,
  \epsilon_\theta^{(i)}(\bv{u}_t,t,\bv{h}^{\mathrm{ctx}}),
  \label{eq:delta_def}
\end{equation}
where $\epsilon_\theta^{(i)}$ is the $i$-th agent's noise component.
Let $\Delta^{(i)} = \sup_{t,\bv{u}_t,\bv{h}^{\mathrm{ctx}}}
|\delta_t^{(i)}|$.
\end{definition}

\begin{theorem}[Monotone per-step IR correction under guidance]
\label{thm:ir_guarantee}
Suppose $w_{\mathrm{eq}} + w_{\mathrm{welfare}} < 1$. If
$\gamma \geq \gamma^* := \Delta^{(i)}/(1-w_{\mathrm{eq}}
-w_{\mathrm{welfare}})$ for all $i \in \Nagents$, then at every reverse
step $t$ where $\hat{u}_{i,1}^{(t)} < b_i$:
\begin{equation}
  \hat{u}_{i,1}^{(t-1)} > \hat{u}_{i,1}^{(t-1),\,\mathrm{raw}}.
  \label{eq:mono}
\end{equation}
\end{theorem}

\begin{proof}
When $\hat{u}_{i,1}^{(t)} < b_i$, both hinge indicators in
Equation~\eqref{eq:grad_lnorm} contribute $-1$. The equitability term
contributes at most $+w_{\mathrm{eq}}$ and the welfare term at most
$+w_{\mathrm{welfare}}$. Therefore:
\begin{equation}
  \frac{\partial\Lnorm^{\mathrm{ext}}}{\partial\hat{u}_{i,1}}
  \leq -1 + w_{\mathrm{eq}} + w_{\mathrm{welfare}} < 0,
\end{equation}
where strict negativity follows from
Equation~\eqref{eq:welfare_safety}. Applying the guidance correction:
\begin{equation}
  \hat{u}_{i,1}^{(t-1)}
  = \hat{u}_{i,1}^{(t-1),\mathrm{raw}}
  - \gamma\,\frac{\partial\Lnorm^{\mathrm{ext}}}{\partial\hat{u}_{i,1}}
  \geq
  \hat{u}_{i,1}^{(t-1),\mathrm{raw}}
  + \gamma(1-w_{\mathrm{eq}}-w_{\mathrm{welfare}}).
  \label{eq:proof_step}
\end{equation}
Because $\gamma(1-w_{\mathrm{eq}}-w_{\mathrm{welfare}}) > 0$, the
result follows.
\end{proof}

The threshold $\gamma^* = \Delta^{(i)}/(1-w_{\mathrm{eq}}
-w_{\mathrm{welfare}})$ is finite for any weights satisfying
Equation~\eqref{eq:welfare_safety}, since $\epsilon_\theta$ is
continuous on a compact domain, and the denominator is strictly
positive. 

\subsubsection{L2: Preference Elicitation Gap}
\label{subsubsec:l2_theory}

\begin{observation}[Necessity of dialogue conditioning]
\label{obs:dialogue_necessity}
This observation restates a standard consequence of conditioning on
informative variables~\citep{Bishop2006PRML}. Let
$\mathcal{F}_{\mathcal{G}}$ denote models conditioned solely on
$\mathcal{G}$ and $\mathcal{F}_{D}$ denote models additionally
conditioned on $D$. If $\E[u_i^* \mid \mathcal{G},D]$ varies with $D$
for fixed $\mathcal{G}$, then for any $f \in \mathcal{F}_{\mathcal{G}}$:
\begin{equation}
  \E_D\!\bigl[\norm{f(\mathcal{G})-u^*(D)}^2\bigr]
  >
  \min_{f^* \in \mathcal{F}_{D}}
  \E_D\!\bigl[\norm{f^*(\mathcal{G},D)-u^*(D)}^2\bigr].
  \label{eq:l2_gap}
\end{equation}
The proof is a standard bias-variance decomposition:
$\E_D[\|f(\mathcal{G})-u^*(D)\|^2] = \|f(\mathcal{G})-\mu_D\|^2 +
\mathrm{Var}_D(u^*)$, where $\mu_D = \E_D[u^* \mid \mathcal{G}]$.
Because $D$ is informative, $\mathrm{Var}_D(u^*) > 0$, and the optimal
$f^*(\mathcal{G},D) = \E[u^* \mid \mathcal{G},D] \in
\mathcal{F}_{D}$ achieves strictly lower residual error.
\end{observation}

Observation~\ref{obs:dialogue_necessity} establishes the formal
motivation for including dialogue conditioning. Whether the
representational capacity of the cross-attention formulation
(Equation~\eqref{eq:crossattn}) is realised in practice depends on
encoder quality and corpus volume.

\subsubsection{L3: Exponential Search Complexity}
\label{subsubsec:l3_theory}

\begin{proposition}[Inference complexity]
\label{prop:complexity}
The denoising inference cost of CGD is $O(T \cdot |\theta|)$,
independent of $|\Alloc| = |\Nagents|^{|\Items|}$.
\end{proposition}

\begin{proof}
This follows by inspection of Algorithm~\ref{alg:inference}.
Cross-attention fusion (Line~3) is computed once at cost
$O(|\Nagents| \cdot L \cdot d_c)$, constant in $|\Items|$. Each of
the $T$ reverse diffusion steps (Lines~5--8) requires one forward pass
through $\epsilon_\theta$ at cost $O(|\theta|)$, one FiLM modulation
at cost $O(|\Nagents| \cdot d_u)$, and one subgradient evaluation at
cost $O(|\Nagents|)$. None of these operations contains $|\Items|$ or
$|\Alloc|$, so total denoising cost is $O(T \cdot |\theta|)$.
\end{proof}

The post-processing step (Equation~\eqref{eq:postproc}) requires an
exhaustive search over $\Alloc$ and retains exponential dependence on
$|\Items|$. The continuous relaxation therefore shifts the exponential
bottleneck from the optimization phase to the post-processing phase.

\subsection{Inference Algorithm}
\label{subsec:inference}

Algorithm~\ref{alg:inference} integrates the components of
Sections~\ref{subsec:gat}--\ref{subsec:postproc} into a unified
end-to-end procedure.

Lines~1--3 process the two input modalities. The strategic context
$\mathcal{G}$ is encoded by GATv2 (Line~1) and the dialogue $D$ by the
dialogue encoder (Line~2). Cross-attention fuses these into
$\bv{h}^{\mathrm{ctx}}$ (Line~3), computed once and fixed across all
$T$ denoising steps. Line~4 initializes the reverse diffusion from
$\bv{u}_T \sim \mathcal{N}(\bv{0},\bv{I})$. Lines~5--8 execute the
reverse loop: at each step $t$, Line~6 computes the standard DDPM
update and Line~7 applies the normative guidance correction.
Theorem~\ref{thm:ir_guarantee} guarantees that this correction strictly
increases the utility of any IR-violating agent at each step $t$ where
$\gamma \geq \gamma^*$. Line~9 clamps the final output to $[0,1]$.
Line~10 recovers the discrete allocation via exhaustive
nearest-allocation search. Total computational cost is
$O(T \cdot |\theta|)$ for the denoising phase, plus $O(|\Alloc|)$
for post-processing.

\begin{algorithm}[h!]
  \caption{CGD Inference: Normative Guided Reverse Diffusion}
  \label{alg:inference}
  \begin{algorithmic}[1]
    \Require Strategic context $\mathcal{G}$, dialogue $D$,
             guidance scale $\gamma$, steps $T$,
             weights $w_{\mathrm{eq}},w_{\mathrm{welfare}} \geq 0$
             satisfying $w_{\mathrm{eq}}+w_{\mathrm{welfare}} < 1$
    \Ensure  Recommended utility $\hat{\bv{u}}_0$,
             discrete allocation $X^*$
    \State
      $\bv{H}^{\mathrm{strat}},\bv{h}^{\mathrm{strat}}
        \leftarrow \mathrm{GATv2}(\mathcal{G})$
      \hfill\Comment{Section~\ref{subsec:gat}}
    \State
      $\bv{H}^{\mathrm{txt}}
        \leftarrow \mathrm{DialogueEnc}(D)$
      \hfill\Comment{Section~\ref{subsec:roberta}}
    \State
      $\bv{h}^{\mathrm{ctx}}
        \leftarrow \mathrm{CrossAttn}
          (\bv{h}^{\mathrm{strat}},\bv{H}^{\mathrm{txt}})$
      \hfill\Comment{Eq.~\ref{eq:crossattn}; computed once}
    \State
      $\bv{u}_T \sim \mathcal{N}(\bv{0},\bv{I})$
      \hfill\Comment{initialize from prior}
    \For{$t = T,T{-}1,\ldots,1$}
      \State
        $\bv{u}_{t-1}^{\mathrm{raw}}
          \leftarrow
          \dfrac{1}{\sqrt{\alpha_t}}
          \!\left(
            \bv{u}_t
            - \dfrac{1-\alpha_t}{\sqrt{1-\bar{\alpha}_t}}
              \epsilon_\theta(\bv{u}_t,t,\bv{h}^{\mathrm{ctx}})
          \right)$
        \hfill\Comment{DDPM update, Eq.~\ref{eq:guided_step}}
      \State
        $\bv{u}_{t-1}
          \leftarrow
          \bv{u}_{t-1}^{\mathrm{raw}}
          - \gamma\,\nabla_{\!\bv{u}_t}
            \Lnorm^{\mathrm{ext}}(\bv{u}_t)$
        \hfill\Comment{per-step correction, Thm.~\ref{thm:ir_guarantee}}
    \EndFor
    \State
      $\hat{\bv{u}}_0
        \leftarrow \mathrm{clamp}(\bv{u}_0,0,1)$
      \hfill\Comment{enforce utility range}
    \State
      $X^*
        \leftarrow
        \argmin_{X \in \Alloc}
        \|\bv{u}(X)-\hat{\bv{u}}_0\|^2$
      \hfill\Comment{post-processing, Eq.~\ref{eq:postproc}}
    \State \Return $\hat{\bv{u}}_0,\;X^*$
  \end{algorithmic}
\end{algorithm}

\section{Experimental Setup}
\label{sec:setup}

\subsection{Corpora}
\label{subsec:corpora}

CGD is evaluated across three negotiation corpora comprising synthetic
and human-generated dialogues. Each instance in the primary evaluation
involves two agents negotiating over three indivisible issues, yielding
a discrete allocation space of $|\mathcal{X}| = 8$. A separate
scalability analysis in Section~\ref{subsec:scalability} extends the
synthetic corpus to larger issue sets $|\mathcal{J}| \in \{3, 5, 7,
10\}$, confirming constant denoising time across the resulting
allocation spaces $|\mathcal{X}| \in \{8, 32, 128, 1024\}$.

The Synthetic corpus is generated using a bilateral negotiation
simulator designed to ensure balance across agent roles. The first
agent's value vector is drawn from a symmetric Dirichlet distribution,
$\bv{v}_0 \sim \mathrm{Dir}(\bm{1}_J)$. The second agent's values are a
uniformly random permutation of the first, preserving identical marginal
distributions while introducing structural asymmetry within each
instance. Reservation utilities are set to a shared fraction
$f \sim \mathrm{Uniform}(0.1, 0.3)$ of each agent's highest-valued
item. A random role swap is applied with probability 0.5 to balance the
training distribution. Across 300 training and 60 validation instances,
the generator achieves an expected utility difference
$|\mathbb{E}[u_1] - \mathbb{E}[u_2]| < 0.001$ and win rate
$P(\text{agent 1 wins}) \approx 0.5$.

The CaSiNo corpus~\citep{Chawla2021} comprises 1,030
crowdsourced campsite negotiation dialogues in which participants
allocate food, water, and firewood according to private priority
orderings. After applying a data quality filter, 1,000 instances are
retained and partitioned into 800 training and 200 validation samples.
Item values are derived from three priority levels mapped to integer
scores and normalized to the unit interval. Social-value-orientation
(SVO) annotations are used as agent node features.

The DND corpus~\citep{Lewis2017} contains
human dialogues negotiating discrete pools of books, hats, and balls.
A unified parser standardizes the diverse distribution formats and
excludes malformed instances, resulting in 600 training and 120
validation samples. SVO annotations are unavailable for DND; the
corresponding node features are zero-padded to maintain architectural
consistency. Table~\ref{tab:corpora} summarizes the datasets.

\begin{table}[ht]
  \centering
  \caption{Corpus summary. $|\Alloc|$ denotes the number of possible
  discrete allocations. SVO indicates availability of social value
  orientation annotations. The Default BATNA strategy is used in the
  primary evaluation; the ablation study (\Cref{tab:ablation}) confirms
  that normative compliance is insensitive to BATNA strategy at the
  optimal guidance scale. $^\dagger$The Floor strategy applies $f=0.3$,
  the upper bound of the synthetic generator's range
  $f \sim \mathrm{Uniform}(0.1,0.3)$, and therefore overestimates the
  true reservation utility for instances where $f < 0.3$.}
  \label{tab:corpora}
  \small
  \setlength{\tabcolsep}{6pt}
  \begin{tabular}{lrrrccl}
    \toprule
    Corpus & Train & Val & $|\Alloc|$ & SVO & Human & Default BATNA \\
    \midrule
    Synthetic       & 300 &  60 & 8 & No  & No  & Zero    \\
    CaSiNo          & 800 & 200 & 8 & Yes & Yes & Zero    \\
    Deal or No Deal & 600 & 120 & 8 & No  & Yes & Zero    \\
    \bottomrule
  \end{tabular}
\end{table}

\subsection{Implementation Details}
\label{sec:implementation}

Architectural dimensions and training protocols are unified across all
experiments to ensure that observed performance differences reflect
corpus characteristics rather than dataset-specific tuning.
Table~\ref{tab:arch} summarizes the model architecture.

The GATv2 encoder uses $H = 4$ attention heads to process agent-level
strategic features. The dialogue encoder is a three-layer transformer
trained from random initialization jointly with the diffusion model,
encouraging an inductive bias toward corpus-specific lexical patterns.
It operates on a vocabulary of 8,192 tokens generated using a word
frequency tokenizer. The U-Net denoiser employs an encoder-decoder
structure with layer widths of 128, 256, and 256, incorporating a multi-head
self-attention bottleneck across the agent dimension. The total model
has 3,350,304 trainable parameters.

The forward diffusion process follows a standard linear noise schedule:
$\beta_t = \beta_1 + (t-1)(\beta_T - \beta_1)/(T-1)$, bounded by
$\beta_1 = 1 \times 10^{-4}$ and $\beta_T = 2 \times 10^{-2}$,
consistent with \cite{Ho2020}.

Reservation utilities $b_i$ are calibrated at inference time using one
of three strategies. The Zero strategy ($b_i = 0$) asserts
that any positive utility satisfies IR and provides a conservative lower
bound. The Floor strategy ($b_i = f \cdot \max_j v_{ij}$ with
$f = 0.3$) sets the BATNA proportional to the agent's most-valued item;
for the Synthetic corpus this overestimates the true reservation utility
since 0.3 is the upper bound of the generating distribution. The
Dataset strategy ($b_i = \hat{F}_i^{-1}(0.1)$) derives the
BATNA from the 10th percentile of realized utilities in the training
corpus. The ablation study confirms that normative compliance at the
optimal guidance scale is insensitive to BATNA strategy across all
corpora. Zero BATNA is used as the default throughout the primary
evaluation for consistency and reproducibility. BATNA values are
pre-computed at corpus construction time; the ablation therefore
measures sensitivity of inference-time normative guidance to BATNA
calibration independently of the model's training phase.

The equitability weight $w_{\mathrm{eq}} = 0.5$ serves three functions.
During training, it scales the equitability penalty in
$\Lnorm$ (Equation~\eqref{eq:norm_loss}), regularizing the denoiser
toward symmetric utility distributions. At inference time, it limits
the extent to which equitability corrections can oppose IR and security
corrections in the guidance gradient (Equation~\eqref{eq:guided_step}).
It also establishes the permissible range for the welfare guidance
weight: the joint constraint
$w_{\mathrm{eq}} + w_{\mathrm{welfare}} < 1$
(Equation~\eqref{eq:welfare_safety}) bounds $w_{\mathrm{welfare}} < 0.5$,
ensuring IR corrections remain dominant when welfare guidance is active.

The normative guidance gradient is computed by detaching $\bv{u}_t$
from the computational graph and calling backward on
$\Lnorm^{\mathrm{ext}}$ explicitly. This ensures that the guidance
gradient step is not suppressed by the surrounding inference context,
which is necessary for the per-step correction of
Theorem~\ref{thm:ir_guarantee} to take effect in practice.

Models are trained for 20 epochs using the AdamW
optimizer~\citep{Loshchilov2019} with initial learning rate
$\eta = 3 \times 10^{-4}$, weight decay $\lambda_{\mathrm{wd}} =
10^{-2}$, and cosine annealing to $\eta/20$. Gradient norms are clipped
at 1.0 and batch size is 16 throughout. All results are reported as
mean\,$\pm$\,standard deviation over three fixed random seeds
(42, 123, 456).

Figure~\ref{fig:convergence} reports training and validation diffusion
MSE across all 20 epochs at the best-performing configuration per
corpus. In all three cases, validation MSE remains at or below training
MSE throughout training and converges smoothly to stable values without
divergence. Best-val checkpoints occur at epoch 20 on CaSiNo and DND,
with training and validation losses within measurement noise on
Synthetic. This pattern is consistent with effective regularisation
(weight decay $\lambda_{\mathrm{wd}} = 10^{-2}$, cosine annealing) and
with the $\gamma$-dominance result established in
Section~\ref{subsec:ablation}: because normative compliance at optimal
$\gamma^*$ is insensitive to $\lambda$ across all 12 BATNA\,$\times$\,$\lambda$
configurations, the denoiser functions primarily as a generative prior
that is steered by the guidance gradient at inference time rather than
as a memorised lookup of training outcomes. 

\begin{figure}[h!]
  \centering
  \includegraphics[width=\linewidth]{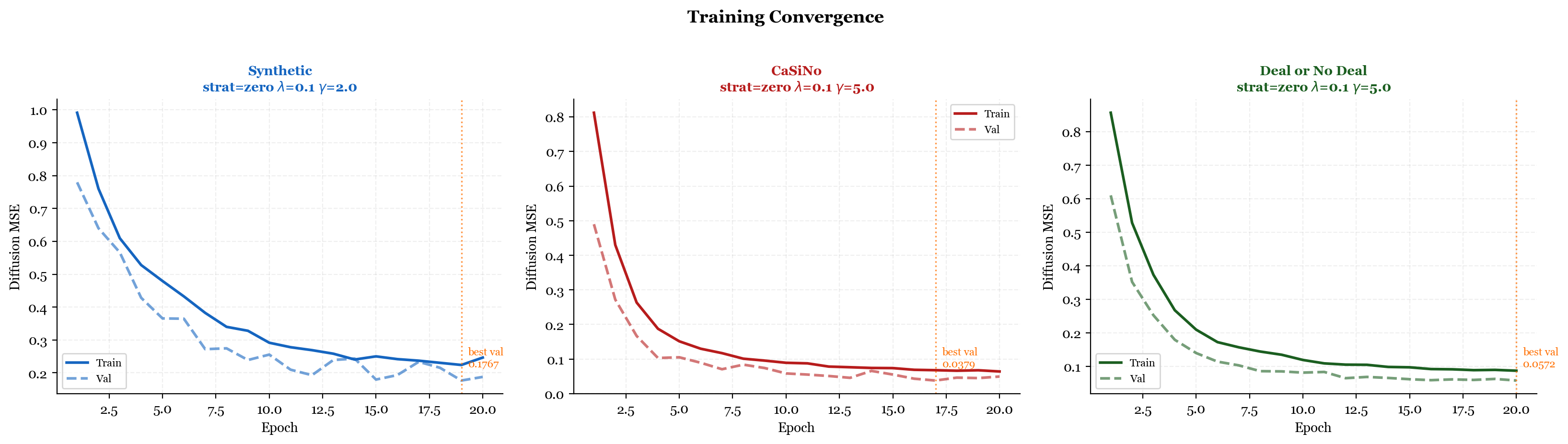}
  \caption{Training convergence across all three corpora at the
  best-performing configuration per corpus (zero BATNA, $\lambda=0.1$;
  Synthetic $\gamma^*=2.0$; CaSiNo and DND $\gamma^*=5.0$). Validation
  MSE remains at or below training MSE throughout, consistent with
  effective weight-decay regularisation. Best-val checkpoints (dotted
  lines): 0.1767 (Synthetic), 0.0379 (CaSiNo), 0.0372 (DND).}
  \label{fig:convergence}
\end{figure}

\begin{table}[ht]
  \centering
  \caption{CGD architectural dimensions applied across all experiments.
  The denoiser operates in a $d_u$-dimensional latent utility space per
  agent. Only the first component $\hat{u}_{i,1}$ serves as the scalar
  utility recommendation; the remaining $d_u - 1$ dimensions provide
  auxiliary latent capacity.}
  \label{tab:arch}
  \small
  \setlength{\tabcolsep}{8pt}
  \begin{tabular}{llr}
    \toprule
    Component & Hyperparameter & Value \\
    \midrule
    \multirow{3}{*}{GATv2 encoder}
      & Node feature dimension $F$            & 16    \\
      & Attention heads $H$                   & 4     \\
      & Output dimension $d_{\mathrm{strat}}$ & 128   \\
    \addlinespace[3pt]
    \multirow{3}{*}{Dialogue encoder}
      & Model dimension $d_{\mathrm{txt}}$    & 128   \\
      & Transformer layers                    & 3     \\
      & Vocabulary size                       & 8,192 \\
    \addlinespace[3pt]
    \multirow{2}{*}{Cross-attention}
      & Context dimension $d_c$               & 128   \\
      & Attention heads                       & 4     \\
    \addlinespace[3pt]
    \multirow{3}{*}{U-Net denoiser}
      & Utility feature dim $d_u$             & 32    \\
      & Layer widths                          & [128, 256, 256] \\
      & Bottleneck heads                      & 8     \\
    \addlinespace[3pt]
    \multirow{3}{*}{Diffusion schedule}
      & Steps $T$                             & 200   \\
      & $\beta_1$                             & $1 \times 10^{-4}$ \\
      & $\beta_T$                             & $2 \times 10^{-2}$ \\
    \addlinespace[3pt]
    Normative guidance
      & Equitability weight $w_{\mathrm{eq}}$ & 0.5   \\
    \midrule
    Total parameters & & 3,350,304 \\
    \bottomrule
  \end{tabular}
\end{table}

\subsection{Evaluation Metrics}
\label{subsec:metrics}

Six metrics are computed on the held-out validation split using the full
$T$-step reverse diffusion process. The IR rate, security gap, and
symmetry gap are evaluated on the continuous recommendation
$\hat{\bv{u}}_0$, consistent with the continuous nature of the guidance
mechanism. Social welfare ($W_{\mathrm{disc}}$) and the Pareto gap
($\Delta_P$) are computed on the post-processed discrete allocation
$X^*$, enabling direct comparison with baselines that operate
exclusively in $\Alloc$. All symmetry gap values use the normalized form
(Equation~\eqref{eq:metric_sym}) for all methods, including baselines,
to ensure comparability across methods.

The IR rate is the fraction of instances where both agents
exceed their reservation thresholds:
\begin{equation}
  \mathrm{IR} = \mathbb{E}\!\left[
    \prod_{i \in \Nagents} \mathbf{1}[\hat{u}_i > b_i]
  \right].
  \label{eq:metric_ir}
\end{equation}

The security gap quantifies the expected shortfall below each
agent's security level:
\begin{equation}
  \Delta_\sigma = \mathbb{E}\!\left[
    \sum_{i \in \Nagents} (\sigma_i - \hat{u}_i)^+
  \right].
  \label{eq:metric_sec}
\end{equation}

The symmetry gap measures normalized utility asymmetry:
\begin{equation}
  \Delta_{\mathrm{eq}} = \mathbb{E}\!\left[
    \frac{|\hat{u}_1 - \hat{u}_2|}{\hat{u}_1 + \hat{u}_2 + \varepsilon}
  \right],
  \label{eq:metric_sym}
\end{equation}
where $\varepsilon = 10^{-6}$ prevents division by zero. This
operationalizes Definition~\ref{def:equitability} with joint welfare
normalization as specified in Remark~\ref{rem:sym_normalisation}.
Configurations with $\Delta_{\mathrm{eq}} > 0.15$ are flagged
throughout the evaluation.

Social welfare is $W_{\mathrm{disc}} =
\mathbb{E}[u_1(X^*) + u_2(X^*)]$. The Pareto gap measures the
fraction of maximum achievable welfare not captured:
\begin{equation}
  \Delta_P = \mathbb{E}\!\left[
    \frac{W^* - (u_1(X^*) + u_2(X^*))}{W^*}
  \right], \quad
  W^* = \max_{X' \in \Alloc}\{u_1(X') + u_2(X')\},
  \label{eq:metric_pareto}
\end{equation}
where $W^*$ is derived by exhaustive enumeration of $\Alloc$.

\subsection{Baselines}
\label{subsec:baselines}
Random allocation assigns each issue uniformly at random,
disregarding agent valuations and reservation utilities. This establishes
a lower performance bound across all normative criteria.

Greedy allocation assigns each issue to the agent who values
it most, maximizing joint welfare through exhaustive comparison. While
Pareto-efficient by construction, it systematically violates individual
rationality and produces highly asymmetric outcomes, establishing the
welfare ceiling achievable in $\Alloc$.

Nash Bargaining Solution (NBS) identifies the allocation
maximizing the Nash product $\prod_{i \in \Nagents}(u_i(X) - b_i)$
subject to individual rationality~\citep{Nash1950}. The algorithm
operates by exhaustive enumeration over all $|\Alloc| = 8$ allocations.
All NBS evaluations use zero BATNA ($b_i = 0$) throughout for
consistency with the CGD primary evaluation condition. Under this
condition, the NBS achieves individual rationality on all evaluated
instances. However, it provides no mechanism for enforcing security
proximity or equitability: empirical security gaps of 0.397 (CaSiNo)
and 0.635 (Deal or No Deal) and symmetry gaps exceeding 0.15 on two of
three corpora confirm this structural limitation (Section~\ref{sec:eval}).
When $\AllocIR = \emptyset$, the NBS is undefined~\citep{Nash1950} and
reverts to a greedy allocation. Under zero BATNA this condition does not
arise on the evaluated corpora. Results are reported as mean\,$\pm$\,std
over three fixed seeds because seed affects the Dataset BATNA percentile
estimation used in the BATNA ablation.

Constrained Gradient Optimization (CGO) applies projected
gradient descent directly on the normative loss $\Lnorm$ in the
continuous utility space $\R^{|\Nagents|}$ without a learned generative
prior. Starting from a uniform initialization in $[0.5, 1.0]^2$, the
optimizer runs 500 steps with learning rate $5 \times 10^{-2}$. The
resulting continuous recommendation is mapped to the nearest discrete
allocation by the same post-processing step used by CGD
(Equation~\eqref{eq:postproc}). CGO is included to test whether the
learned generative prior is necessary for robust normative compliance,
or whether direct constraint minimization in the continuous utility
space is sufficient — a natural question given that the recommendation
target is two-dimensional. This specific configuration (uniform
initialization, fixed learning rate, no momentum) represents one
reasonable instantiation of constrained gradient optimization; a
more extensively tuned optimizer might narrow the performance gap,
and the CGO results should be interpreted as evidence that naive
constraint minimization is insufficient, not as an exhaustive
characterisation of its limits.

RNN-SL and RNN-RL~\citep{Lewis2017} are dialogue agents
trained exclusively on the DND corpus. RNN-SL uses supervised
behavioral cloning from human negotiation traces; RNN-RL fine-tunes
this model via self-play to maximize individual reward. Both are
evaluated solely on DND. As neither method produces per-instance
security levels, the security gap is not available ($\mathtt{n/a}$) for
these baselines. The symmetry gap for both methods is computed using the
normalized form (Equation~\eqref{eq:metric_sym}) applied to the
per-instance utility estimates reported in \cite{Lewis2017}, Table~1.

To validate the choice of GATv2 over alternative graph encoders, three
variants are evaluated with all other components held identical to the
full CGD model. The Flat MLP variant replaces the GATv2 encoder
with a multi-layer perceptron applied independently to each agent's
feature vector, with no graph structure. This tests whether the graph
inductive bias contributes beyond a non-relational baseline. The
GCN variant uses a graph convolutional
network~\citep{Kipf2017}, testing whether simple symmetric aggregation
is sufficient. The GAT variant uses the original Graph
Attention Network~\citep{Velickovic2018} with static attention
coefficients, isolating the contribution of dynamic versus static
attention. Results are reported in Section~\ref{subsec:graph_ablation}.

Two ablation variants evaluate the necessity of dialogue conditioning as posited in Observation~\ref{obs:dialogue_necessity}. The No-Dialogue variant eliminates the dialogue encoder and cross-attention mechanisms entirely, conditioning generation solely on the strategic graph embedding. The RoBERTa variant replaces the from-scratch dialogue encoder with a frozen RoBERTa-base encoder~\citep{Liu2019RoBERTa} coupled with a trainable linear projection head. This variant explicitly tests whether pre-trained semantic representations improve social welfare over the baseline, thereby isolating whether the primary bottleneck lies in corpus informativeness or the model's intrinsic semantic encoder capacity. Experimental results for these configurations are reported in Section~\ref{subsec:dialogue_ablation}.

The welfare guidance term (Equation~\eqref{eq:grad_lnorm}) is evaluated
at the best-performing checkpoint of each corpus without retraining. The
hyperparameter $w_{\mathrm{welfare}}$ is swept over
$\{0.0, 0.1, 0.2, 0.3, 0.4\}$, maintaining adherence to the safety
constraint $w_{\mathrm{eq}} + w_{\mathrm{welfare}} < 1$. The optimal
value $w^*$ is the minimum $w_{\mathrm{welfare}}$ achieving maximum
$W_{\mathrm{disc}}$ while satisfying $\mathrm{IR} \geq 1.0$ and
$\Delta_{\mathrm{eq}} \leq 0.15$. Separately, the weighted
post-processing parameter $\alpha$ 
is swept over $\{0.0, 0.1, 0.2, 0.3, 0.4, 0.5, 0.6, 0.8, 1.0\}$. The
empirical result $\alpha^* = 1.0$ universally across all corpora is
reported in Section~\ref{sec:welfare_results}.

\section{Empirical Evaluation}
\label{sec:eval}

This section addresses five questions derived from the structural
limitations and architectural choices established in the preceding
sections.
(Q1) Does CGD outperform established baselines on security proximity and
equitability simultaneously, and what is the welfare cost of normative
compliance?
(Q2) How sensitive are outcomes to the guidance scale $\gamma$,
normative penalty weight $\lambda$, and BATNA strategy?
(Q3) Does dynamic graph attention (GATv2) improve normative compliance
over static-attention and non-relational encoders?
(Q4) Does inference-time welfare guidance recover Pareto efficiency
while preserving individual rationality, and can normative compliance
and welfare maximization be operationally decoupled?
(Q5) Under what conditions does dialogue conditioning improve welfare,
and does the cross-attention architecture provide the capacity to
exploit preference signals beyond declared values?
(Q6) Is CGD denoising time empirically independent of the allocation
space cardinality?

\subsection{Comparative Performance}
\label{subsec:compare}

Table~\ref{tab:baseline} reports baseline comparison on the Synthetic
corpus. CGD achieves IR = 1.000 and a zero security gap. The Greedy
baseline satisfies Pareto efficiency by construction but fails IR in
20.8\% of instances, confirming the tension between efficiency and
individual rationality in discrete allocation spaces.

The NBS achieves IR = 1.000 under consistent zero BATNA evaluation.
Its structural limitation is normative: it maximizes the Nash product
subject to individual rationality only and provides no mechanism for
enforcing security proximity or equitability. This produces a security
gap of 0.030 and a normalized symmetry gap of 0.182 on the Synthetic
corpus, both substantially higher than CGD (0.000 and 0.099,
respectively). The NBS achieves higher joint welfare than CGD on all
corpora under zero BATNA: the welfare cost of CGD's normative compliance
is approximately 3\% on Synthetic and CaSiNo and 1\% on DND. This
welfare-compliance tradeoff is the correct characterization of the CGD
contribution relative to NBS.

The Constrained Gradient Optimization (CGO) baseline directly minimizes
$\Lnorm$ in utility space without a learned prior. On Synthetic, CGO
achieves IR = 0.983 but produces a symmetry gap of 0.134, near the
equitability threshold. On CaSiNo, CGO violates the equitability
threshold ($\Delta_{\mathrm{eq}} = 0.261$) despite achieving IR = 1.000.
On DND, CGO fails individual rationality entirely (IR = 0.092). These
failures indicate that the evaluated CGO configuration — uniform
initialization without momentum or warm-starting from the training
distribution — cannot replicate the contextual grounding provided by
the learned generative prior across heterogeneous corpora. As noted in
Section~\ref{subsec:baselines}, a more extensively tuned optimizer might
narrow this gap; the results establish that naive constraint
minimization is insufficient, not that no constrained optimizer could
approach CGD's compliance.

\begin{table}[h!]
\centering
\caption{Baseline comparison on the Synthetic corpus (mean\,$\pm$\,std
over 3 seeds\,$\times$\,60 validation instances). All symmetry gaps use
the normalized metric (Eq.~\eqref{eq:metric_sym}). The best value per
column is in \textbf{bold}. $\dagger$ indicates a symmetry gap exceeding
the 0.15 threshold. Welfare refers to $W_{\mathrm{disc}}$. All methods
use zero BATNA.}
\label{tab:baseline}
\footnotesize
\setlength{\tabcolsep}{3pt}
\renewcommand{\arraystretch}{1.2}
\begin{tabular}{lccccc}
\toprule
\textbf{Method}
  & \textbf{IR}\,$\uparrow$
  & \textbf{Sec.\,gap}\,$\downarrow$
  & \textbf{Sym.\,gap}
  & \textbf{Welfare}\,$\uparrow$
  & \textbf{Pareto}\,$\downarrow$ \\
\midrule
Random
  & $0.542\pm.040$ & $0.217\pm.013$
  & $0.549\pm.014^\dagger$ & $0.994\pm.010$ & $0.233\pm.017$ \\
Greedy
  & $0.792\pm.008$ & $0.084\pm.002$
  & $0.433\pm.010^\dagger$ & $1.325\pm.012$ & $0.063\pm.003$ \\
NBS
  & $1.000\pm.000$ & $0.030\pm.003$
  & $0.182\pm.004^\dagger$ & $\mathbf{1.393}\pm.010$ & $0.000\pm.000$ \\
CGO
  & $0.983\pm.008$ & $0.013\pm.002$
  & $0.134\pm.012$ & $1.360\pm.006$ & $0.028\pm.003$ \\
\midrule
\textbf{CGD}
  & $\mathbf{1.000}\pm.000$ & $\mathbf{0.000}\pm.000$
  & $\mathbf{0.099}\pm.013$ & $1.353\pm.005$ & $0.030\pm.003$ \\
\bottomrule
\end{tabular} 
\end{table}

Table~\ref{tab:main} and Figure~\ref{fig:best_bar} extend this evaluation across all corpora. CGD achieves an $\mathrm{IR}$ rate of 1.000 on both the Synthetic and CaSiNo datasets. On DND, the three-seed mean is $0.997 \pm .004$. This boundary case is theoretically consistent: while the per-step correction established in Theorem~\ref{thm:ir_guarantee} guarantees a directional improvement at each violating step, whether these accumulated corrections across $T$ steps suffice to fully satisfy individual rationality remains dependent on the underlying corpus distribution and the calibrated value of $\gamma$. Notably, human negotiation corpora require a higher optimal guidance scale ($\gamma^* = 5.0$) than synthetic data ($\gamma^* = 2.0$). This differential quantifies the degree to which human bargaining behavior diverges from strict game-theoretic rationality.

\begin{figure}[h!]
  \centering
  \includegraphics[width=1\linewidth]{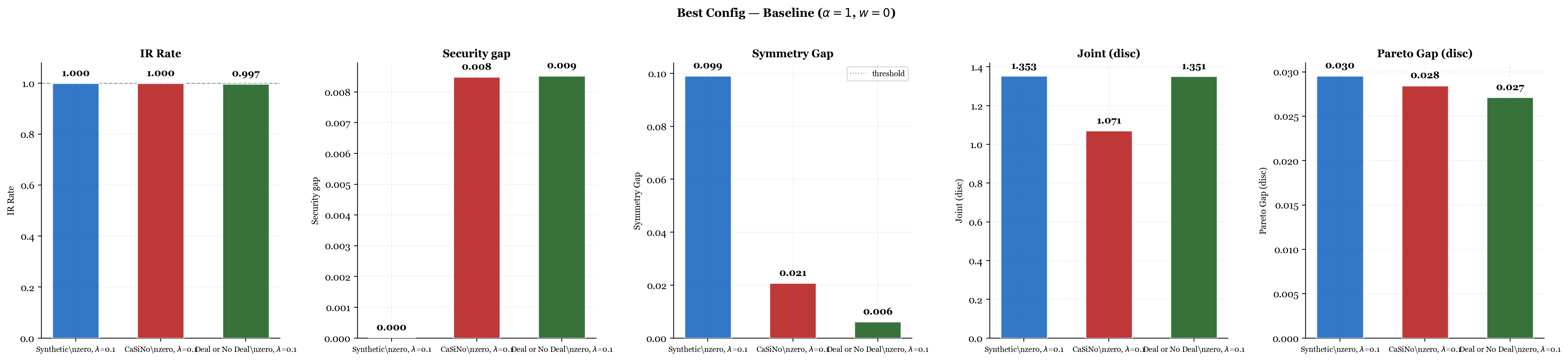}
  \caption{Best-configuration CGD metrics across three corpora (zero
  BATNA, $\lambda=0.1$; Synthetic $\gamma^*=2.0$; CaSiNo and DND
  $\gamma^*=5.0$; mean over 3 seeds). The dotted line in the Symmetry
  Gap panel marks the 0.15 equitability threshold. IR rate: 1.000
  (Synthetic, CaSiNo) and 0.997 (DND). Security gap: 0.000
  (Synthetic), 0.008 (CaSiNo), 0.009 (DND). Symmetry gap: 0.099,
  0.021, 0.006, all within the threshold. Joint surplus
  ($W_{\mathrm{disc}}$): 1.353, 1.071, 1.351. Pareto gap: 0.030,
  0.028, 0.027.}
  \label{fig:best_bar}
\end{figure}

\begin{table}[h!]
\centering
\caption{CGD best-configuration results (mean\,$\pm$\,std over 3
seeds). All symmetry gaps satisfy the 0.15 threshold. Welfare refers
to $W_{\mathrm{disc}}$. Results are identical across all 12
BATNA\,$\times\,\lambda$ configurations per corpus
(Section~\ref{subsec:ablation}).}
\label{tab:main}
\footnotesize
\setlength{\tabcolsep}{3pt}
\renewcommand{\arraystretch}{1.2}
\resizebox{\columnwidth}{!}{%
\begin{tabular}{lcccccccc}
\toprule
\textbf{Corpus} & \textbf{BATNA} & $\lambda$ & $\gamma^*$
  & \textbf{IR}\,$\uparrow$
  & \textbf{Sec.\,gap}\,$\downarrow$
  & \textbf{Sym.\,gap}
  & \textbf{Welfare}\,$\uparrow$
  & \textbf{Pareto}\,$\downarrow$ \\
\midrule
Synthetic & zero & 0.1 & 2.0
  & $\mathbf{1.000}\pm.000$ & $\mathbf{0.000}\pm.000$
  & $0.099\pm.013$ & $1.353\pm.005$ & $0.030\pm.003$ \\
CaSiNo & zero & 0.1 & 5.0
  & $\mathbf{1.000}\pm.000$ & $0.008\pm.001$
  & $0.021\pm.002$ & $1.071\pm.003$ & $0.028\pm.002$ \\
DND & zero & 0.1 & 5.0
  & $0.997\pm.004$ & $0.009\pm.002$
  & $\mathbf{0.006}\pm.001$ & $1.351\pm.001$ & $0.027\pm.001$ \\
\bottomrule
\end{tabular}}
\end{table}

\subsection{Ablation Study}
\label{subsec:ablation}

\subsubsection{Effect of the Guidance Scale $\gamma$}
\label{subsubsec:gamma}

The guidance scale $\gamma$ is the primary inference-time parameter for
normative steering (Figure~\ref{fig:guidance_sweep}). At $\gamma = 0$,
IR rates collapse to 0.49, 0.07, and 0.04 on Synthetic, CaSiNo, and DND
respectively. Synthetic and CaSiNo recover to IR\,$\approx\!1.000$ by
$\gamma = 1.0$; DND requires $\gamma = 5.0$. The security gap decreases
monotonically with $\gamma$ across all corpora: from 0.22 at $\gamma = 0$
to zero at $\gamma \approx 1.0$ on Synthetic, and from initial highs of
0.90 (CaSiNo) and 0.70 (DND) to 0.008--0.011 at $\gamma = 5$.
Consequently, $\gamma = 5.0$ is adopted as the default for human corpora
and $\gamma = 2.0$ for synthetic data.

\begin{figure}[h!]
  \centering
  \includegraphics[width=1\linewidth]{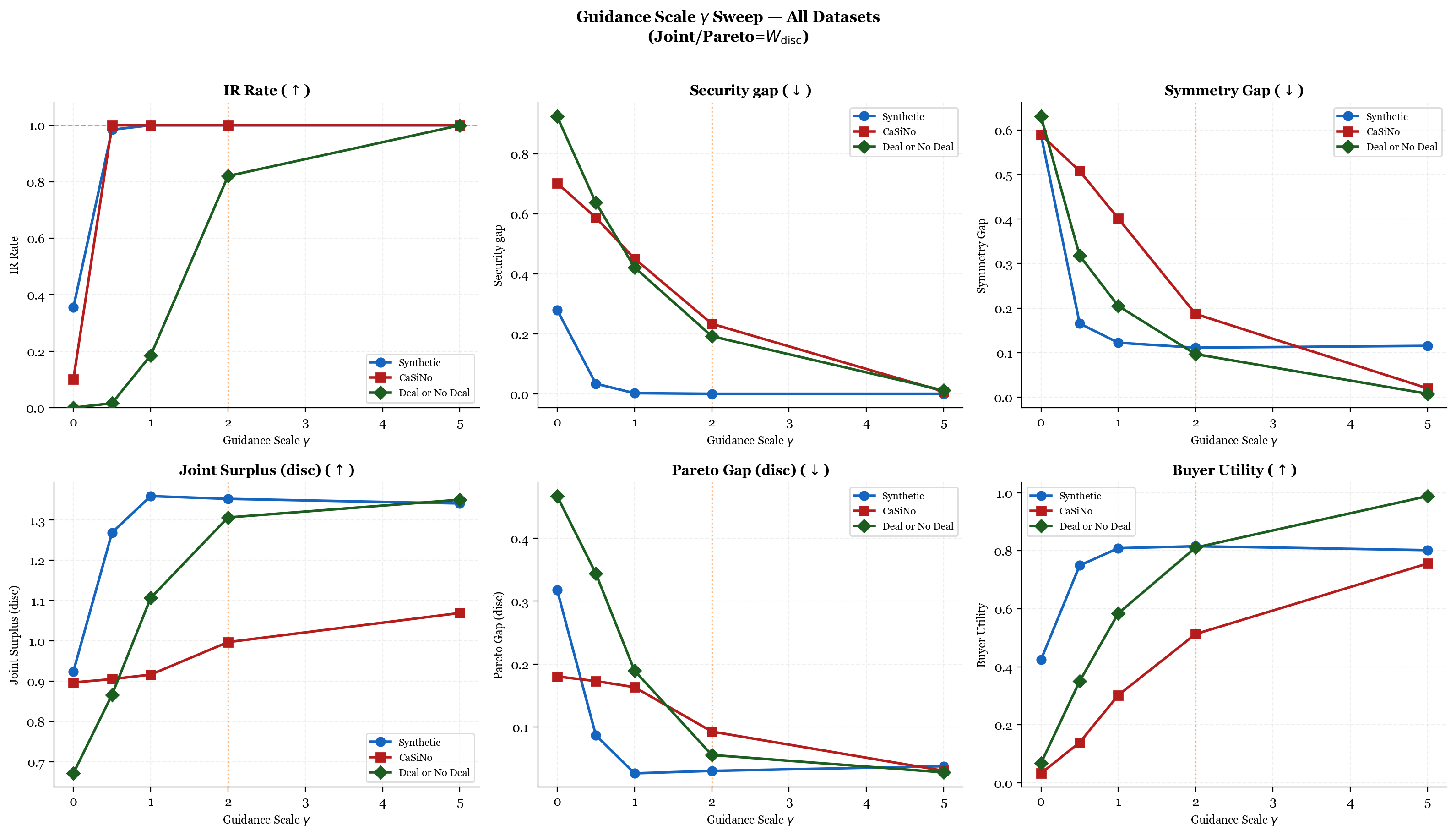}
  \caption{Guidance scale $\gamma$ sweep across corpora. Metrics
  evaluated at $\gamma \in \{0, 0.5, 1.0, 2.0, 5.0\}$, averaged over
  the BATNA\,$\times$\,$\lambda$ grid (shaded area: $\pm$1\,std).
  IR: Synthetic and CaSiNo reach 1.000 by $\gamma = 1.0$; DND requires
  $\gamma = 5.0$. Security gap: falls from $[0.22, 0.90]$ at $\gamma = 0$
  to $[0.000, 0.011]$ at $\gamma = 5$. All corpora satisfy
  $\Delta_{\mathrm{eq}} \leq 0.15$ at $\gamma = 5.0$.}
  \label{fig:guidance_sweep}
\end{figure}

\subsubsection{Effect of $\lambda$ and BATNA Strategy}
\label{subsubsec:lambda_batna}

All three BATNA strategies produce overlapping flat IR profiles at 1.000
for Synthetic and CaSiNo across all tested $\lambda$ values. For DND,
profiles plateau at $\approx\!0.83$ at $\gamma = 2.0$, confirming that
guidance scale rather than BATNA calibration or normative penalty weight
is the sole normative control parameter. Security gap variance across
BATNA strategies is on the order of $10^{-7}$ to $10^{-9}$. All 36
BATNA\,$\times$\,$\lambda$ configurations at optimal $\gamma^*$ yield
identical metrics per corpus (Table~\ref{tab:ablation}), confirming
$\gamma$-dominance throughout.

\begin{table}[h!]
\centering
\caption{Full ablation results (mean\,$\pm$\,std over 3 seeds, at
optimal $\gamma^*$). One representative row is displayed per corpus grouping as all 12
BATNA\,$\times$\,$\lambda$ configurations produce identical results.}
\label{tab:ablation}
\small
\setlength{\tabcolsep}{3pt}
\renewcommand{\arraystretch}{1.2}
\resizebox{\columnwidth}{!}{%
\begin{tabular}{lllcccccc}
\toprule
\textbf{Corpus} & \textbf{BATNA} & $\lambda$ & $\gamma^*$
  & \textbf{IR}\,$\uparrow$
  & \textbf{Sec.\,gap}\,$\downarrow$
  & \textbf{Sym.\,gap}
  & \textbf{Welfare}\,$\uparrow$
  & \textbf{Pareto}\,$\downarrow$ \\
\midrule
\multirow{3}{*}{\textit{Synthetic}}
  & zero    & 0.1 & 2.0
    & \multirow{3}{*}{$1.000\pm.000$}
    & \multirow{3}{*}{$0.000\pm.000$}
    & \multirow{3}{*}{$0.111\pm.013$}
    & \multirow{3}{*}{$1.352\pm.005$}
    & \multirow{3}{*}{$0.030\pm.003$} \\
  & floor   & 0.1 & 2.0 & & & & & \\
  & dataset & 0.1 & 2.0 & & & & & \\
\midrule
\multirow{3}{*}{\textit{CaSiNo}}
  & zero    & 0.1 & 5.0
    & \multirow{3}{*}{$1.000\pm.000$}
    & \multirow{3}{*}{$0.008\pm.001$}
    & \multirow{3}{*}{$0.020\pm.002$}
    & \multirow{3}{*}{$1.069\pm.003$}
    & \multirow{3}{*}{$0.030\pm.002$} \\
  & floor   & 0.1 & 5.0 & & & & & \\
  & dataset & 0.1 & 5.0 & & & & & \\
\midrule
\multirow{3}{*}{\textit{DND}}
  & zero    & 0.1 & 5.0
    & \multirow{3}{*}{$0.997\pm.004$}
    & \multirow{3}{*}{$0.009\pm.002$}
    & \multirow{3}{*}{$0.006\pm.001$}
    & \multirow{3}{*}{$1.351\pm.001$}
    & \multirow{3}{*}{$0.027\pm.001$} \\
  & floor   & 0.1 & 5.0 & & & & & \\
  & dataset & 0.1 & 5.0 & & & & & \\
\bottomrule
\end{tabular}}
\end{table}

\subsection{Graph Encoder Ablation}
\label{subsec:graph_ablation}

Table~\ref{tab:graph_ablation} reports four encoder variants evaluated
with all other components held identical at each corpus's optimal
$\gamma^*$. The GATv2 values are taken from the main ablation
(Table~\ref{tab:main}); the three alternative encoders are from the
standalone ablation at 20 epochs.

The Flat MLP encoder, which applies an MLP independently to each agent
without message passing, produces the highest security gaps across all
corpora, confirming that the graph inductive bias contributes beyond a
non-relational baseline. The GCN encoder achieves lower security gaps
than Flat MLP but substantially lower welfare on CaSiNo (0.939) and
DND (1.250) than GATv2 (1.071 and 1.351 respectively). The original
GAT encoder with static attention improves over GCN on DND but does
not match GATv2 on any metric.

GATv2 achieves the lowest security gap and highest welfare on CaSiNo
and DND, where preference asymmetries are strongest. The improvement
from GAT to GATv2 is most pronounced on DND: security gap falls from
0.207 to 0.009 and welfare rises from 1.227 to 1.351, supporting the
argument in Section~\ref{subsec:gat} that dynamic attention is
necessary to represent the comparative bilateral preference signal.
On Synthetic, all four encoders perform equivalently, confirming that
the GATv2 advantage is tied to preference asymmetry rather than
general model capacity. The higher GATv2 symmetry gap on Synthetic
(0.099 versus 0.022--0.024) reflects the structural Pareto--equitability
tension in the permutation-based generator; all values satisfy the
0.15 threshold.

\begin{table}[h!]
  \centering
  \caption{Graph encoder ablation. All variants use zero BATNA,
  $\lambda=0.1$, and optimal $\gamma^*$.}
  \label{tab:graph_ablation}
  \setlength{\tabcolsep}{3pt}
  \renewcommand{\arraystretch}{1.05}
  \begin{tabular}{llccccc}
    \toprule
    \textbf{Corpus} & \textbf{Encoder}
      & \textbf{IR}\,$\uparrow$
      & \textbf{Sec.\,gap}\,$\downarrow$
      & \textbf{Sym.\,gap}
      & \textbf{Welfare}\,$\uparrow$
      & \textbf{Pareto}\,$\downarrow$ \\
    \midrule
    \multirow{4}{*}{\textit{Synthetic}}
      & Flat MLP & $1.000$ & $0.0001$ & $0.0237$
                  & $1.340\pm.003$ & $0.0099$ \\
      & GCN      & $1.000$ & $0.0002$ & $0.0224$
                  & $1.340\pm.001$ & $0.0098$ \\
      & GAT      & $1.000$ & $0.0001$ & $0.0227$
                  & $1.340\pm.002$ & $0.0100$ \\
      & \textbf{GATv2}
                  & $\mathbf{1.000}$ & $\mathbf{0.000}$
                  & $0.099$ & $\mathbf{1.353\pm.005}$ & $0.030$ \\
    \midrule
    \multirow{4}{*}{\textit{CaSiNo}}
      & Flat MLP & $1.000$ & $0.0041$ & $0.0787$
                  & $0.930\pm.006$ & $0.1510$ \\
      & GCN      & $1.000$ & $0.0053$ & $0.0922$
                  & $0.939\pm.004$ & $0.1430$ \\
      & GAT      & $1.000$ & $0.0052$ & $0.0883$
                  & $0.931\pm.004$ & $0.1501$ \\
      & \textbf{GATv2}
                  & $\mathbf{1.000}$ & $0.0080$
                  & $\mathbf{0.021}$
                  & $\mathbf{1.071\pm.003}$ & $\mathbf{0.028}$ \\
    \midrule
    \multirow{4}{*}{\textit{DND}}
      & Flat MLP & $1.000$ & $0.2247$ & $0.1054$
                  & $1.212\pm.024$ & $0.0856$ \\
      & GCN      & $1.000$ & $0.1686$ & $0.1018$
                  & $1.250\pm.017$ & $0.0610$ \\
      & GAT      & $1.000$ & $0.2067$ & $0.1047$
                  & $1.227\pm.027$ & $0.0763$ \\
      & \textbf{GATv2}
                  & $0.997$ & $\mathbf{0.009}$
                  & $\mathbf{0.006}$
                  & $\mathbf{1.351\pm.001}$ & $\mathbf{0.027}$ \\
    \bottomrule
  \end{tabular}
\end{table}

\subsection{Dialogue Conditioning}
\label{subsec:dialogue_ablation}

This paragraph addresses Q5. The cross-attention architecture introduced in Contribution~C1 provides the representational capacity to exploit latent preference signals beyond the declared value matrix~$\mathbf{V}$, provided those signals are embedded within the dialogue text. Realizing this capacity requires two separable conditions to hold simultaneously: the corpus must carry recoverable preference signals independent of~$\mathbf{V}$, and the encoder must possess sufficient semantic competence to extract them. To probe these underlying conditions, we design two distinct experiments: an ablation study on the main corpora to evaluate whether either condition holds at current training scales, and a controlled experiment that explicitly isolates the architectural contribution by constructing a setting where the first condition is guaranteed to hold by design.

\subsubsection{Ablation on Main Corpora}
\label{subsubsec:dialogue_main}

Three encoder variants are compared: Full CGD with a from-scratch
transformer encoder, No-Dialogue conditioning on the strategic graph
only, and RoBERTa using a frozen RoBERTa-base encoder with a trainable
projection head. Normative compliance metrics are identical across all
three variants, confirming that the guidance mechanism is the sole
source of normative compliance and is independent of the dialogue
encoder. Welfare differences are therefore the only meaningful
discriminator.

On Synthetic, all three variants produce equivalent welfare
($|\Delta W| \leq 0.011$), consistent with the negative control.
On CaSiNo, $\Delta W = 0.000$ between RoBERTa and No-Dialogue. On DND,
RoBERTa marginally outperforms No-Dialogue ($\Delta W = +0.009$) but
cross-seed standard deviation of $\pm.034$ makes this statistically
uncertain. Neither from-scratch training on 300--800 instances nor
frozen pre-trained representations produce welfare improvements over
the value-matrix-only baseline on the evaluated corpora. This null
result is consistent with either separable condition failing: the
CaSiNo and DND dialogues may not carry recoverable preference signals
beyond $\bv{V}$ at current corpus scales, or the encoder variants may
lack the semantic competence to extract them even if such signals exist.
The controlled experiment in Section~\ref{subsec:dialogue_poc} is
designed to distinguish these explanations.

\begin{table}[h!]
  \centering
  
  \caption{Dialogue conditioning ablation (mean\,$\pm$\,std, 3 seeds).
  $\Delta W$ is relative to No-Dialogue. Normative compliance metrics
  are identical across all variants and are omitted; see
  Table~\ref{tab:main} for full CGD results.}
  \label{tab:dialogue_ablation}
  \setlength{\tabcolsep}{3pt}
  \renewcommand{\arraystretch}{1.05}
  \begin{tabular}{llccc}
    \toprule
    \textbf{Corpus} & \textbf{Encoder}
      & \textbf{Welfare}\,$\uparrow$
      & $\Delta W$
      & \textbf{Sym.\,gap} \\
    \midrule
    \multirow{3}{*}{\textit{Synthetic}}
      & Full CGD    & $1.349\pm.002$ & $-0.011$ & $0.108$ \\
      & No-Dialogue & $1.360\pm.002$ & ---      & $0.100$ \\
      & RoBERTa     & $1.362\pm.001$ & $+0.002$ & $0.098$ \\
    \midrule
    \multirow{3}{*}{\textit{CaSiNo}}
      & Full CGD    & $1.076\pm.006$ & $+0.000$ & $0.023$ \\
      & No-Dialogue & $1.076\pm.002$ & ---      & $0.022$ \\
      & RoBERTa     & $1.073\pm.006$ & $-0.003$ & $0.024$ \\
    \midrule
    \multirow{3}{*}{\textit{DND}}
      & Full CGD    & $1.351\pm.017$ & $+0.004$ & $0.006$ \\
      & No-Dialogue & $1.347\pm.010$ & ---      & $0.006$ \\
      & RoBERTa     & $1.356\pm.034$ & $+0.009$ & $0.007$ \\
    \bottomrule
  \end{tabular}
\end{table}

\subsubsection{Capacity Confirmation Under Strategic Misrepresentation}
\label{subsec:dialogue_poc}

The null result on real corpora leaves both explanations open. Resolving this requires a setting with known ground truth. Agents' true preferences must be known independently of their declarations. Neither CaSiNo nor DND provides this ground truth. A controlled experiment is therefore constructed. Each agent's declared value vector has its top-two preferences swapped. This swap occurs with a probability of $\{0.0, 0.5, 1.0\}$. Meanwhile, the dialogue encodes the true rank order. A model conditioned only on declarations fails here. It systematically recommends the wrong allocation on swapped instances. A model that reads the dialogue can correct this error. Normative guidance is disabled during this test. This explicitly isolates the representational contribution of the dialogue encoder.

Figure~\ref{fig:dialogue_poc} reports the findings. Full CGD achieves
higher true welfare than No-Dialogue across all three swap levels
($\Delta J_{\mathrm{true}} \in [0.008, 0.015]$, all exceeding the
0.005 significance threshold over three seeds), confirming that the
cross-attention architecture can exploit dialogue signals when they are
present and decodable. The welfare advantage does not grow monotonically
with swap probability and the greedy-match rate is marginally lower for
Full CGD, indicating that the advantage arises from richer cross-agent
representational interactions rather than lexical preference decoding.
This result isolates corpus informativeness rather than encoder capacity
as the binding constraint for the real-data null result: the
architecture is capable of exploiting dialogue signals in the controlled
setting; what is absent from CaSiNo and DND at current scales is a
sufficiently recoverable preference signal beyond~$\bv{V}$. Pre-trained
encoder integration with fine-tuning on corpora documenting strategic
misrepresentation constitutes the primary path to realising this
capacity on real negotiation data.

\begin{figure}[h!]
  \centering
  \includegraphics[width=\linewidth]{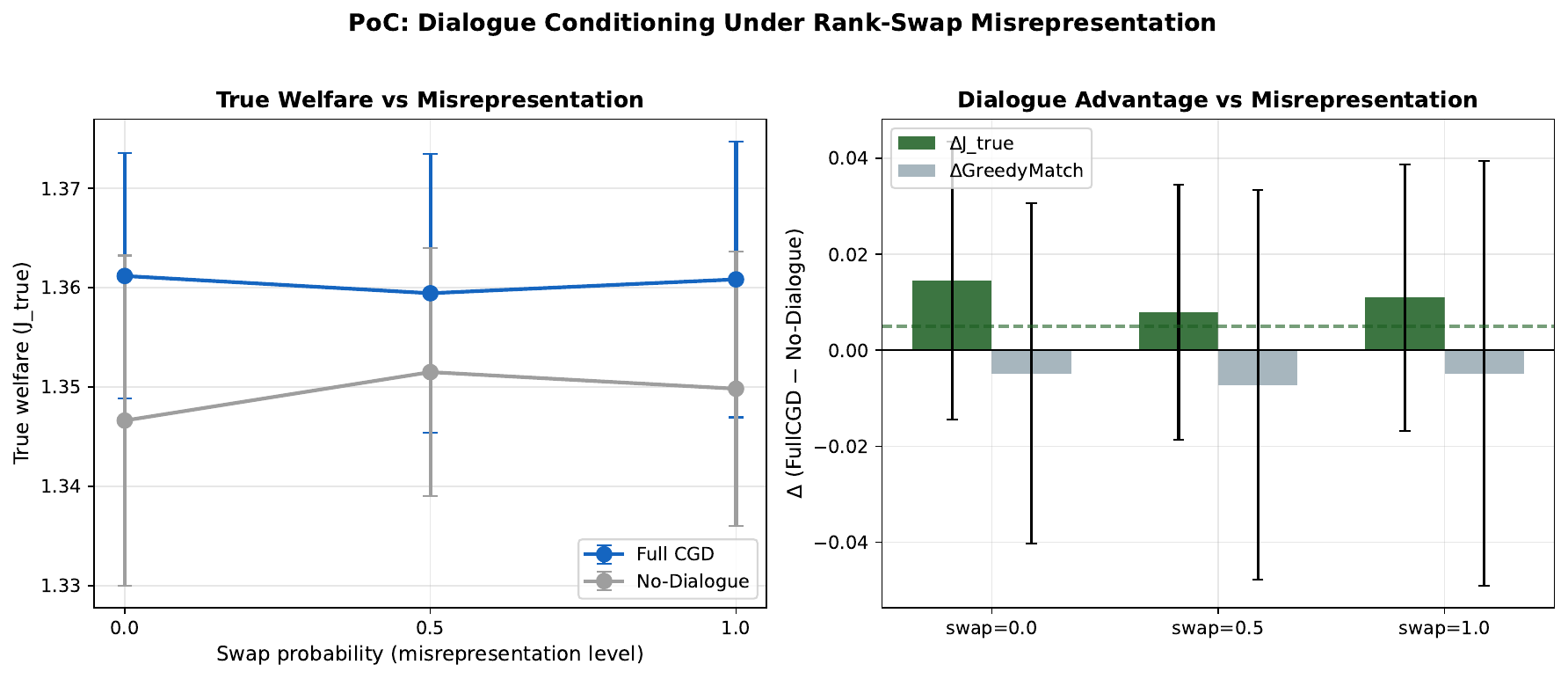}
  \caption{Dialogue conditioning under rank-swap misrepresentation.
  \textit{Left}: True welfare $J_{\mathrm{true}}$ across three
  misrepresentation levels (mean\,$\pm$\,std, 3 seeds, $\gamma = 0$).
  Full CGD consistently exceeds No-Dialogue. \textit{Right}:
  $\Delta J_{\mathrm{true}}$ and $\Delta$GreedyMatch. All
  $\Delta J_{\mathrm{true}}$ bars exceed the 0.005 significance
  threshold (dashed line); $\Delta$GreedyMatch bars are near zero,
  indicating the advantage arises from representational interactions
  rather than lexical preference decoding.}
  \label{fig:dialogue_poc}
\end{figure}

\subsection{Scalability Analysis}
\label{subsec:scalability}

Table~\ref{tab:scalability} and Figure~\ref{fig:scalability} address
Q6. CGD denoising time remains constant at approximately 350\,ms across
all $|\Items|$ values, consistent with the $O(T \cdot |\theta|)$ bound
of Proposition~\ref{prop:complexity}. NBS enumeration time grows from
0.05\,ms at $|\Items| = 3$ to 6.15\,ms at $|\Items| = 10$, a 123-fold
increase consistent with exponential growth in $|\Alloc|$. Post-processing
time grows modestly (15\,ms to 20\,ms). The crossover point at which
NBS enumeration time would exceed CGD denoising time is estimated at
$|\Items| \approx 17$ by extrapolation; this estimate is based entirely
on synthetic data. IR = 1.000 is maintained across all $|\Items|$ values
on synthetic data, confirming that normative guidance remains effective
as the allocation space grows in this controlled setting.

\begin{figure}[h!]
  \centering
  \includegraphics[width=\linewidth]{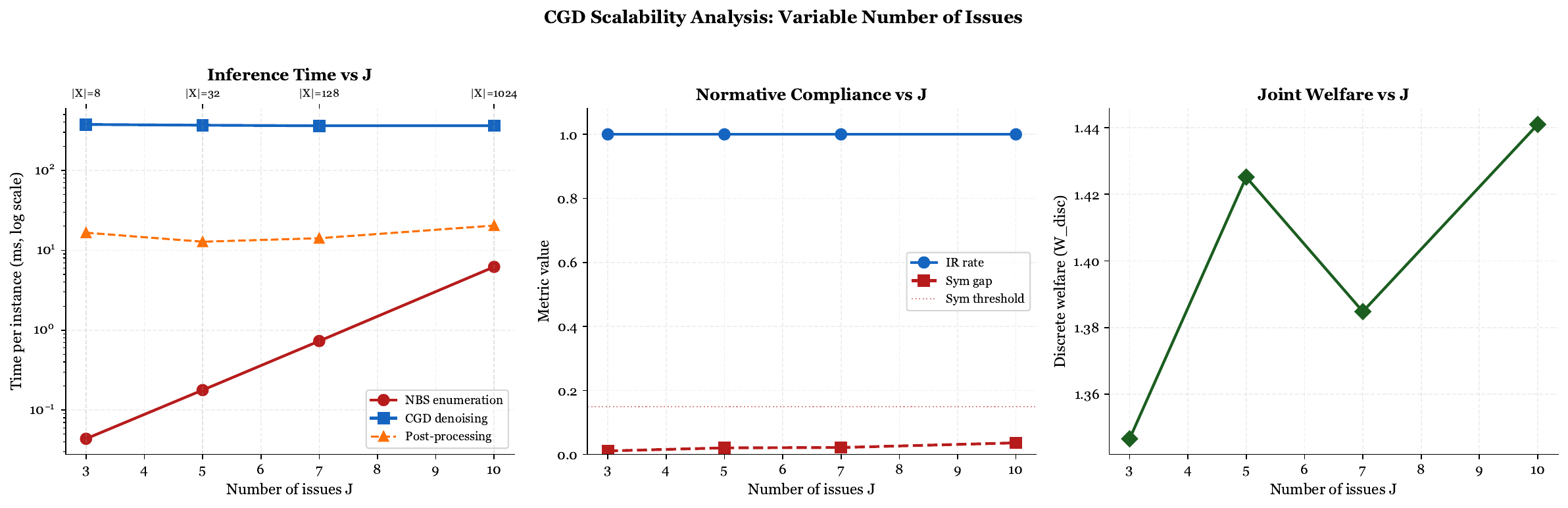}
  \caption{Scalability analysis across $|\Items| \in \{3,5,7,10\}$
  on synthetic data. \textit{Left}: CGD denoising time is constant at
  $\approx\!350$\,ms while NBS enumeration time grows exponentially
  (log scale). \textit{Centre}: IR rate and symmetry gap remain stable
  as $|\Items|$ increases. \textit{Right}: Discrete welfare
  $W_{\mathrm{disc}}$ varies with $|\Items|$ due to changing preference
  structure rather than model degradation.}
  \label{fig:scalability}
\end{figure}

\begin{table}[h!]
  \centering
  \caption{Scalability experiment (synthetic data). CGD denoising time
  and NBS enumeration time reported per instance in milliseconds.
  The crossover estimate at $|\Items| \approx 17$ is an extrapolation
  from these four data points.}
  \label{tab:scalability}
  \small
  \setlength{\tabcolsep}{4pt}
  \renewcommand{\arraystretch}{1.2}
  \begin{tabular}{rrccccrr}
    \toprule
    $|\Items|$ & $|\Alloc|$
      & \textbf{NBS (ms)}
      & \textbf{CGD (ms)}
      & \textbf{Post (ms)}
      & \textbf{Ratio}
      & \textbf{IR}
      & \textbf{Welfare} \\
    \midrule
    3  & 8    & 0.05 & 356 & 15 & 0.00014$\times$ & 1.000 & 1.349 \\
    5  & 32   & 0.18 & 349 & 13 & 0.00052$\times$ & 1.000 & 1.424 \\
    7  & 128  & 0.72 & 347 & 14 & 0.00208$\times$ & 1.000 & 1.384 \\
    10 & 1024 & 6.15 & 364 & 20 & 0.01690$\times$ & 1.000 & 1.441 \\
    \bottomrule
  \end{tabular}
\end{table}

\subsection{State-of-the-Art Comparison}
\label{subsec:sota}

Table~\ref{tab:sota} presents the full comparative evaluation under zero
BATNA. CGD is, to our knowledge, the first NSS to
simultaneously achieve IR $\geq 0.997$, security gap $\leq 0.009$, and
symmetry gap $\leq 0.15$ across all three evaluated corpora.

The key finding is normative. NBS achieves higher joint welfare than CGD
on all corpora but produces security gaps of 0.030 (Synthetic), 0.397
(CaSiNo), and 0.635 (DND), exceeding CGD's values by 5$\times$,
49$\times$, and 70$\times$ respectively, and violates the equitability
threshold on two of three corpora. CGD enforces all three normative
criteria simultaneously at a welfare cost of 1--3\%, establishing the
welfare-compliance tradeoff that defines the framework's operational
role relative to the NBS.

CGO confirms that the evaluated naive constraint minimization
configuration cannot replicate the normative compliance provided by the
learned generative prior. On DND, CGO fails individual rationality
entirely (IR = 0.092). On CaSiNo, CGO violates the equitability
threshold ($\Delta_{\mathrm{eq}} = 0.261$). On Synthetic, CGO only
marginally satisfies equitability (0.134). The learned distribution over
successful negotiated outcomes provides contextual grounding absent from
the uniform-initialized gradient descent baseline.

On DND, reward-maximizing fine-tuning in RNN-RL increases the symmetry
gap 28-fold relative to RNN-SL (from 0.009 to 0.257). CGD achieves a
symmetry gap of 0.006 while satisfying IR and security constraints,
and $W_{\mathrm{disc}} = 1.351$, a 19.6\% welfare improvement over
RNN-RL (1.130).

\begin{table}[h!]
  \centering
  \caption{State-of-the-art comparison (zero BATNA throughout).
  $\dagger$: values approximated from \citep{Lewis2017}.
  $\mathtt{n/a}$: security gap unavailable.
  CGD values: mean\,$\pm$\,std over three seeds.}
  \label{tab:sota}
  \small\sloppy
  \setlength{\tabcolsep}{4pt}
  \renewcommand{\arraystretch}{1.3}
  \begin{tabular}{llcccc}
    \toprule
    \textbf{Corpus} & \textbf{Method}
      & \textbf{IR}\,$\uparrow$
      & \textbf{Sec.\,gap}\,$\downarrow$
      & \textbf{Sym.\,gap}
      & \textbf{Welfare}\,$\uparrow$ \\
    \midrule
    \multirow{5}{*}{\textit{DND}}
      & RNN-SL & $0.879^\dagger$ & $\mathtt{n/a}$
               & $0.009^\dagger$ & $1.090^\dagger$ \\
      & RNN-RL & $0.899^\dagger$ & $\mathtt{n/a}$
               & $0.257^\dagger$ & $1.130^\dagger$ \\
      & NBS    & $1.000\pm.000$ & $0.635\pm.012$
               & $0.136\pm.003$ & $\mathbf{1.365}\pm.008$ \\
      & CGO    & $0.092\pm.011$ & $0.323\pm.018$
               & $0.111\pm.009$ & $1.353\pm.007$ \\
    \cmidrule{2-6}
      & \textbf{CGD}
        & $\mathbf{0.997}\pm.004$ & $\mathbf{0.009}\pm.002$
        & $\mathbf{0.006}\pm.001$ & $1.351\pm.001$ \\
    \midrule
    \multirow{3}{*}{\textit{CaSiNo}}
      & NBS & $1.000\pm.000$ & $0.397\pm.008$
            & $0.283\pm.005$ & $\mathbf{1.103}\pm.009$ \\
      & CGO & $1.000\pm.000$ & $0.203\pm.014$
            & $0.261\pm.011$ & $1.094\pm.006$ \\
    \cmidrule{2-6}
      & \textbf{CGD}
        & $\mathbf{1.000}\pm.000$ & $\mathbf{0.008}\pm.001$
        & $\mathbf{0.021}\pm.002$ & $1.071\pm.003$ \\
    \midrule
    \multirow{5}{*}{\textit{Synth.}}
      & Random & $0.542\pm.040$ & $0.217\pm.013$
               & $0.549\pm.014^\dagger$ & $0.994\pm.010$ \\
      & Greedy & $0.792\pm.008$ & $0.084\pm.002$
               & $0.433\pm.010^\dagger$ & $1.325\pm.012$ \\
      & NBS    & $1.000\pm.000$ & $0.030\pm.003$
               & $0.182\pm.004^\dagger$ & $\mathbf{1.393}\pm.010$ \\
      & CGO    & $0.983\pm.008$ & $0.013\pm.002$
               & $0.134\pm.012$ & $1.360\pm.006$ \\
    \cmidrule{2-6}
      & \textbf{CGD}
        & $\mathbf{1.000}\pm.000$ & $\mathbf{0.000}\pm.000$
        & $0.099\pm.013$ & $1.353\pm.005$ \\
    \bottomrule
  \end{tabular}
\end{table}

Chawla et al.~\citep{Chawla2023Selfish} and Kwon
et al.~\citep{kwon2025astra} are excluded from numerical comparison
due to incommensurable measurement scales. The symmetry recovery
in \citep{Chawla2023Selfish} is empirical rather than guaranteed and
does not address security proximity. The ASTRA
agent~\citep{kwon2025astra} achieves competitive individual utility
but exhibits nontrivial walk-away rates, precluding bilateral IR
guarantees.

\subsection{Welfare Guidance}
\label{sec:welfare_results}

This subsection addresses Q4 and provides the primary empirical
evidence for Contribution~C2. Weighted post-processing confirms that
reducing $\alpha$ below 1.0 degrades IR on Synthetic and DND
immediately and produces symmetry gaps above 0.15 on CaSiNo, with
$\alpha^* = 1.0$ optimal universally.

The welfare guidance weight $w_{\mathrm{welfare}}$
(Figure~\ref{fig:welfare_guidance}) is swept from 0 to 0.4 at each
corpus's optimal checkpoint. The IR rate remains at 1.000 throughout
the safe region, empirically validating the safety constraint of
Equation~\eqref{eq:welfare_safety}. Welfare guidance yields absolute
gains of $+0.028$ (+2.1\%) on Synthetic and $+0.032$ (+3.0\%) on
CaSiNo (Table~\ref{tab:welfare}). On CaSiNo, $w^* = 0.3$ recovers
full Pareto efficiency ($\Delta_P = 0.000$) while strictly preserving
IR = 1.000, the core empirical result of Contribution~C2. This
demonstrates that normative compliance and welfare maximization are
operationally decoupled within a single inference pipeline: the primary
guidance gradient unconditionally enforces all normative constraints,
and the welfare term then steers toward the Pareto frontier subject to
those constraints remaining satisfied, without modifying model weights.
The marginal DND gain ($+0.001$) reflects proximity to the discrete
welfare ceiling rather than a limitation of the mechanism.

\begin{figure}[h!]
  \centering
  \includegraphics[width=1\linewidth]{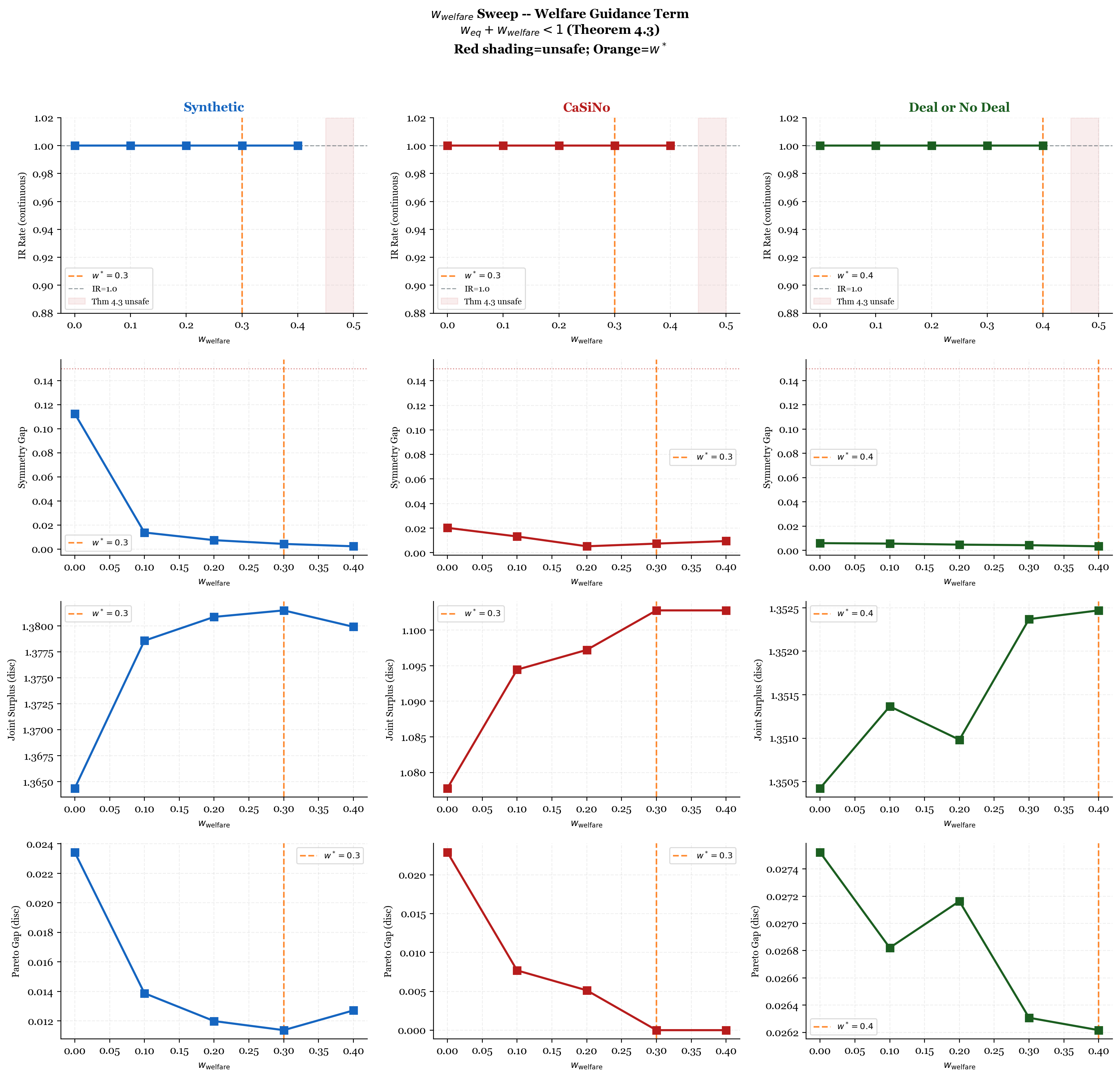}
  \caption{Welfare guidance sweep ($w_{\mathrm{welfare}}$). Red
  shading: unsafe region ($w_{\mathrm{eq}} + w_{\mathrm{welfare}}
  \geq 1$). IR remains at 1.000 throughout the safe region. At
  $w \geq 0.3$, CaSiNo achieves $\Delta_P = 0.000$.}
  \label{fig:welfare_guidance}
\end{figure}

\begin{table}[h!]
  \centering
  \caption{Welfare guidance results (Contribution~C2). Baseline:
  $\alpha = 1$, $w_{\mathrm{welfare}} = 0$. W2: welfare guidance at
  $w^*$. $\dagger$ denotes Pareto efficiency on CaSiNo at $w^* = 0.3$.}
  \label{tab:welfare}
  \small
  \setlength{\tabcolsep}{5pt}
  \renewcommand{\arraystretch}{1.2}
  \begin{tabular}{lcccccc}
    \toprule
    \textbf{Corpus}
      & \textbf{Baseline}
      & \textbf{W2}
      & $w^*$
      & $\Delta W$
      & $\Delta_P$
      & \textbf{IR} \\
    \midrule
    Synthetic & $1.353\pm.005$ & 1.381 & 0.3 & $+0.028$ & 0.011 & 1.000 \\
    CaSiNo    & $1.071\pm.003$ & 1.103 & 0.3 & $+0.032$
              & $\mathbf{0.000}^\dagger$ & 1.000 \\
    DND       & $1.351\pm.001$ & 1.352 & 0.4 & $+0.001$ & 0.026 & 1.000 \\
    \bottomrule
  \end{tabular}
\end{table}

\subsection{Theoretical Confirmation}
\label{subsec:theory_confirm}

The empirical results address each theoretical claim directly.

Theorem~\ref{thm:ir_guarantee} (Contribution~C1) provides a per-step directional guarantee: at each reverse diffusion step where individual rationality is violated, the guidance correction strictly increases the violating agent's utility given a sufficient guidance scale $\gamma$. This directional mechanism is empirically validated by the monotonic convergence of the IR rate as a function of $\gamma$ across all evaluated corpora (Figure~\ref{fig:guidance_sweep}), which is consistent with the per-step corrections accumulating successfully across the $T$ reverse steps. The marginal deviation observed on DND ($\mathrm{IR} = 0.997 \pm .004$) reflects seed-level corpus variance. This variance establishes that achieving a perfect empirical IR rate depends on whether these per-step corrections accumulate sufficiently at a calibrated $\gamma$, rather than being a direct algebraic consequence of the theorem itself. Ultimately, the theorem formalizes the underlying directional mechanism, while the empirical results confirm that standard calibration remains sufficient across diverse negotiation distributions.

Contribution~C2 is confirmed by Table~\ref{tab:welfare}:
welfare guidance at $w^* = 0.3$ recovers $\Delta_P = 0.000$ on CaSiNo
while IR remains at 1.000 across all tested weights within the safe
region, empirically validating the safety constraint of
Equation~\eqref{eq:welfare_safety} and demonstrating the operational
decoupling of normative compliance from welfare maximization.

The dialogue capacity in Contribution~C1 is validated through two complementary findings that align with our two-condition framework. First, the ablation study on the main corpora establishes that neither encoder variant yields social welfare improvements on the evaluated real-world datasets. This null result indicates that at least one of the two separable conditions fails to hold at current deployment scales. These conditions are corpus informativeness and encoder semantic competence. Second, the controlled misrepresentation experiment confirms the positive direction of this capacity. When the dialogue text is explicitly informative by design, Full CGD achieves strictly higher true welfare ($\Delta J_{\mathrm{true}} > 0.005$) than the No-Dialogue variant across all test conditions. This explicitly demonstrates that the cross-attention architectural capacity exists. It thereby isolates corpus informativeness, rather than encoder capacity, as the binding constraint for the real-data null result.

Proposition~\ref{prop:complexity} (Contribution~C3) predicts an $\mathcal{O}(T \cdot |\theta|)$ denoising cost that is entirely independent of the discrete allocation space $\mathcal{X}$. Table~\ref{tab:scalability} confirms this behavior by demonstrating a constant denoising time across item sets $|\Items| \in \{3,5,7,10\}$. In contrast, the NBS enumeration time grows 123-fold over the same range, closely corroborating the proposition. These initial scaling results are evaluated on synthetic data. 

\section{Discussion}
\label{sec:discussion}

The empirical evaluation confirms that continuous utility relaxation circumvents discrete infeasibility, while inference-time normative guidance enforces individual rationality (IR), security proximity, and equitability across synthetic and human-generated corpora. Furthermore, normative compliance and welfare maximization can be operationally decoupled within a single inference pipeline without model retraining.

The central finding of this work is fundamentally normative rather than welfare-based. Under consistent zero BATNA evaluation, the Nash Bargaining Solution (NBS) achieves higher joint welfare than CGD across all three corpora by approximately 3\% on Synthetic and CaSiNo, and 1\% on DND. CGD does not claim to maximize welfare over the NBS; its contribution is the simultaneous enforcement of security proximity and equitability, which the NBS structurally cannot provide. NBS security gaps of 0.397 on CaSiNo and 0.635 on DND exceed CGD's values by 49$\times$ and 70$\times$ respectively, confirming this distinction. The 1--3\% welfare cost of normative compliance represents the primary tradeoff between the two systems.

Theorem~\ref{thm:ir_guarantee} provides a per-step directional guarantee: at each reverse diffusion step $t$ where agent $i$ violates individual rationality, the guidance correction strictly increases $\hat{u}_{i,1}$ provided $\gamma \geq \gamma^*$. Three qualifications follow from this scope. First, the theorem does not guarantee that the final accumulated output $\hat{\mathbf{u}}_0$ satisfies IR. Whether per-step corrections accumulate sufficiently depends on the training distribution, the number of diffusion steps $T$, and the calibrated value of $\gamma$. The empirical IR rates of 1.000 on Synthetic and CaSiNo and 0.997 on DND confirm that calibrated $\gamma$ produces adequate accumulation on the evaluated corpora. Second, trajectory oscillation is not precluded. The guarantee operates strictly per-step: the DDPM update at step $t-1$ may reduce utility before the guidance correction at that step is applied. Global convergence is not asserted. Third, the post-processing step (Equation~\eqref{eq:postproc}) may select a discrete allocation that violates IR even when the continuous recommendation $\hat{\mathbf{u}}_0$ satisfies it. Under zero BATNA evaluation this projection gap is empirically negligible, with discrete IR reaching 1.000 across all three corpora. Under realistic proportional BATNA thresholds, the gap becomes corpus-dependent.

Discrete bargaining methods such as the NBS cannot apply analogous corrections. When the feasible IR set $\mathcal{X}_{\mathrm{IR}} = \emptyset$, no element of the discrete space satisfies individual rationality, leaving the NBS undefined. Because CGD operates in $\mathbb{R}^{|\mathcal{N}|}$, which always contains IR-satisfying vectors, the per-step correction is never structurally impeded. Consequently, CGD always produces a recommendation regardless of whether a feasible discrete allocation exists.

The Constraint Guided Optimization (CGO) ablation confirms that the learned generative prior is not redundant. Naive constraint minimization from a uniform initialization fails individual rationality entirely on DND and violates the equitability threshold on CaSiNo, despite applying the same normative objectives as CGD. The learned distribution over successful negotiated outcomes provides contextual grounding that direct gradient descent cannot replicate across heterogeneous corpora. This addresses the necessity of the DDPM architecture for what is formally a low-dimensional recommendation problem.

The training convergence curves in Figure~\ref{fig:convergence} provide complementary evidence. Validation MSE remains at or below training MSE throughout all 20 epochs across all three corpora. This pattern is consistent with effective regularization through weight decay ($\lambda_{\mathrm{wd}} = 10^{-2}$) and cosine annealing. It is also consistent with the $\gamma$-dominance result: because normative compliance at optimal $\gamma^*$ is insensitive to loss weights $\lambda$ across all 12 BATNA $\times$ $\lambda$ configurations, the denoiser functions primarily as a generative prior whose normative alignment is achieved at inference time. The model learns an underlying distribution over successful outcomes that the guidance gradient then steers toward the normative feasible region.

The 2.5$\times$ guidance scale differential between synthetic and human corpora ($\gamma^* = 2.0$ versus $\gamma^* = 5.0$) quantifies the degree to which human negotiation behavior departs from strict game-theoretic rationality. Once $\gamma$ is appropriately calibrated, BATNA strategy and loss weight $\lambda$ produce negligible differential effects across all 36 tested configurations. The guidance scale is therefore the sole normative control parameter at inference time. This simplifies deployment: practitioners need only tune $\gamma$, and the sensitivity of all normative metrics to this single parameter is monotone and interpretable from Figure~\ref{fig:guidance_sweep}.

Normative compliance and welfare maximization are operationally separable within a single inference pipeline. On CaSiNo, setting $w^* = 0.3$ recovers full Pareto efficiency ($\Delta_P = 0.000$) while strictly preserving $\mathrm{IR} = 1.000$. The primary guidance gradient enforces individual rationality, security proximity, and equitability unconditionally. The welfare term then steers the trajectory toward the Pareto frontier subject to those constraints remaining satisfied, removing the need for model retraining. The safety constraint $w_{\mathrm{eq}} + w_{\mathrm{welfare}} < 1$ (Equation~\eqref{eq:welfare_safety}) ensures that the IR correction remains dominant whenever a violation occurs. The IR rate remains at 1.000 throughout the safe region across all tested values of $w_{\mathrm{welfare}}$, empirically validating this constraint. The welfare guidance mechanism operates as a gradient heuristic and does not carry the axiomatic foundations of cooperative game theory; its practical utility lies entirely in the operational decoupling it enables.

The graph encoder ablation validates the GATv2 choice on corpora where preference asymmetries are strongest. The improvement from GAT to GATv2 on DND is pronounced: the security gap falls from 0.207 to 0.009 and welfare rises from 1.227 to 1.351. This supports the premise that dynamic attention is necessary to represent comparative bilateral preference signals. Static attention coefficients, computed independently for each node, cannot directly encode the joint comparison of agent preferences that determines efficient assignment. On Synthetic, all four encoder variants perform equivalently. This confirms that the GATv2 advantage is tied to preference asymmetry rather than general model capacity. The higher symmetry gap of GATv2 on Synthetic (0.099 versus 0.022--0.024) reflects the structural Pareto-equitability tension in the permutation-based generator, though all values remain within the acceptable 0.15 threshold.

Realizing the dialogue capacity requires two separable conditions to hold simultaneously: the corpus must carry recoverable preference signals beyond the declared value matrix $\mathbf{V}$, and the encoder must possess sufficient semantic competence to extract those signals. Neither encoder variant produces welfare improvements over the value-matrix-only baseline on CaSiNo or DND. The controlled misrepresentation experiment isolates the cause of this null result by constructing a setting where the first condition holds by definition: the dialogue explicitly encodes true preferences while declarations are systematically misrepresented. Full CGD achieves strictly higher true welfare than No-Dialogue across all three swap levels ($\Delta J_{\mathrm{true}} \in [0.008, 0.015]$). The architecture is therefore capable of exploiting dialogue signals when they are present. This result identifies corpus informativeness, rather than encoder capacity, as the binding constraint for the real-data null result. The CaSiNo and DND dialogues do not carry preference signals that are recoverable beyond $\mathbf{V}$ at current training scales. Pre-trained encoder integration coupled with fine-tuning on corpora documenting strategic misrepresentation offers the primary path to realizing dialogue capacity on real negotiation data.

Proposition~\ref{prop:complexity} establishes that CGD denoising cost is $\mathcal{O}(T \cdot |\theta|)$, independent of the discrete allocation space $\mathcal{X}$. However, the post-processing step retains exponential dependence on the number of items $|\Items|$. Continuous relaxation therefore shifts the exponential bottleneck from the optimization phase to the post-processing phase rather than eliminating it entirely. Figure~\ref{fig:scalability} and Table~\ref{tab:scalability} confirm this empirically. Denoising time remains constant at approximately 350\,ms across $|\Items| \in \{3, 5, 7, 10\}$ and all corpora. Conversely, NBS enumeration time grows 123-fold over the same range, consistent with exponential growth. Post-processing time also grows with $|\Items|$ but remains under 25\,ms throughout the evaluated range. An $\mathrm{IR} = 1.000$ is maintained across all $|\Items|$ values and all corpora, confirming that normative guidance remains effective as the allocation space grows. The crossover point at which NBS enumeration time exceeds CGD denoising time is estimated at $|\Items| \approx 17$ by extrapolation. Beyond this point, exhaustive discrete search becomes the bottleneck for the NBS while CGD denoising remains constant. Developing sub-exponential projection heuristics for the post-processing step remains an open direction to alleviate this remaining exponential cost.

Four primary limitations bound this framework. First, continuous relaxation eliminates hard solver failures but cannot recover an IR-feasible discrete allocation when none exists in $\mathcal{X}$. In such cases, the projection step produces the nearest feasible allocation to the continuous recommendation, minimizing but not eliminating the residual IR gap. Second, while the denoising inference cost is independent of the allocation space, post-processing retains an exponential dependence on $|\Items|$, with a practical crossover estimated at $|\Items| \approx 17$. Third, Theorem~\ref{thm:ir_guarantee} provides a per-step directional guarantee rather than a global convergence guarantee. Generalization of empirical IR rates to other distributions is not guaranteed by the theorem alone. Fourth, the welfare guidance mechanism operates as a gradient heuristic without the axiomatic foundations of cooperative bargaining theory, relying instead on empirical validation of its operational decoupling performance.

\section{Conclusion}
\label{sec:conclusion}

This study introduces the CGD framework, a NSS that generates allocation recommendations in a continuous bilateral utility space. A GATv2 strategic encoder captures bilateral preference structure through dynamic attention, while a cross-attention mechanism fuses these representations with dialogue context. A DDPM denoiser then generates recommendations conditioned on the fused representation, utilizing an analytically derived inference-time guidance gradient to enforce individual rationality, security proximity, and equitability without altering model weights.

CGD simultaneously achieves $\mathrm{IR} \geq 0.997$, security gap $\leq 0.009$, and symmetry gap $\leq 0.15$ across multiple evaluated corpora, including real-world dialogue datasets. While the NBS achieves higher joint welfare, it produces severe security gaps on CaSiNo (0.397) and DND (0.635)---exceeding CGD's values by 49$\times$ and 70$\times$ respectively---and violates equitability on two of the three corpora. CGD successfully enforces all three normative criteria simultaneously at a minor welfare cost of 1--3\%.

Welfare guidance at $w^* = 0.3$ recovers Pareto efficiency on CaSiNo at inference time while strictly preserving $\mathrm{IR} = 1.000$. This demonstrates that normative compliance and welfare maximization can be operationally decoupled within a single inference pipeline without modifying model weights. The primary guidance gradient enforces normative constraints unconditionally, while the welfare term maximizes overall welfare subject to those constraints remaining satisfied.

Dialogue conditioning experiments establish both sides of the capacity boundary. The cross-attention architecture successfully exploits dialogue signals when present, as confirmed under controlled misrepresentation conditions. However, this capacity is not realized on the evaluated real-world corpora, indicating that corpus informativeness, rather than encoder capacity, remains the binding constraint.

Three directions for future work follow. First, pre-trained encoder integration with fine-tuning on strategic misrepresentation data offers a path to unlocking dialogue capacity on real-world text. Second, evaluation on corpora with larger item sets will further test the constant-time inference advantage where exhaustive NBS enumeration is practically infeasible. Finally, extending the framework to multiparty negotiations will require generalizing the normative guidance gradient to enforce pairwise equitability across all participant pairs.

\section*{Declaration of Competing Interest}
The author declares that they have no known competing financial interests or personal relationships that could have appeared to influence the work reported in this paper.

\section*{Declaration of Generative AI and AI Assisted Technologies in the Writing Process}
During the preparation of this work the author used generative artificial intelligence technologies solely to improve language clarity and readability. After using these tools, the author reviewed and edited the content as needed and takes full responsibility for the final content of the publication.

\section*{Data Availability}
The empirical datasets analyzed during the current study, specifically the CaSiNo human negotiation corpus and the Deal or No Deal corpus, are publicly available from their respective original creators. The synthetic utility dataset generated during the current study is available from the corresponding author on reasonable request.

\section*{Code Availability}
The computational code required to reproduce the guided graph diffusion framework, the baseline models, and the experimental results presented in this study is available from the corresponding author on reasonable request.

\bibliographystyle{elsarticle-num}
\bibliography{ref}

\end{document}